\documentclass[11pt,letterpaper]{article}

\usepackage{graphicx}
\usepackage[margin=1in]{geometry}
\usepackage{enumitem}
\usepackage{mleftright}
\usepackage{dsfont}
\usepackage{mathtools}
\usepackage{amsmath,amssymb,amsthm}
\usepackage{thmtools}
\usepackage{mathabx}
\usepackage{bm}
\usepackage{nicefrac}       
\usepackage{microtype}      
\usepackage{comment}
\usepackage{bbm}
\usepackage{aligned-overset}

\makeatletter
\def\namedlabel#1#2{\begingroup
    #2%
    \def\@currentlabel{#2}%
    \phantomsection\label{#1}\endgroup
}
\makeatother

\PassOptionsToPackage{hyphens}{url}\usepackage[backref=page]{hyperref}
\hypersetup{
    colorlinks,
    citecolor=blue,
    filecolor=blue,
    linkcolor=blue,
    urlcolor=blue
}
\renewcommand*{\backref}[1]{}
\renewcommand*{\backrefalt}[4]{%
    \ifcase #1 %
    \or     #2%
    \else   #2%
    \fi
}
\usepackage{listings}
\usepackage{fancyhdr}

\usepackage[T1]{fontenc}
\usepackage{framed,color}
\usepackage{lmodern}
\usepackage{bigints}
\usepackage{cancel}
\allowdisplaybreaks

\usepackage{algorithm}
\usepackage{algpseudocode}

\let\originalleft\left
\let\originalright\right
\renewcommand{\left}{\mathopen{}\mathclose\bgroup\originalleft}
\renewcommand{\right}{\aftergroup\egroup\originalright}

\definecolor{codegreen}{rgb}{0,0.6,0}
\definecolor{codegray}{rgb}{0.5,0.5,0.5}
\definecolor{codepurple}{rgb}{0.58,0,0.82}
\definecolor{backcolour}{rgb}{0.95,0.95,0.92}

\lstdefinestyle{mystyle}{
    backgroundcolor=\color{backcolour},   
    commentstyle=\color{codegreen},
    keywordstyle=\color{magenta},
    numberstyle=\tiny\color{codegray},
    stringstyle=\color{codepurple},
    basicstyle=\footnotesize,
    breakatwhitespace=false,         
    breaklines=true,                 
    captionpos=b,                    
    keepspaces=true,                 
    numbers=left,                    
    numbersep=5pt,                  
    showspaces=false,                
    showstringspaces=false,
    showtabs=false,                  
    tabsize=2
}
 
\makeatletter
\newcommand{\AlgoResetCount}{\renewcommand{\@ResetCounterIfNeeded}{\setcounter{AlgoLine}{0}}}
\newcommand{\AlgoNoResetCount}{\renewcommand{\@ResetCounterIfNeeded}{}}
\newcounter{AlgoSavedLineCount}

\makeatother

\usepackage{tikz}
\tikzset{every picture/.style={line width=0.75pt}} 

\newcommand{\N}{{\rm I\!N}}

\newcommand{\bI}{\mathbb{I}}

\newcommand{\bR}{\mathbb{R}}
\newcommand{\bS}{\mathbb{S}}

\newcommand{\id}{\mathbb{I}}

\newcommand{\cD}{\mathcal{D}}
\newcommand{\cE}{\mathcal{E}}

\newcommand{\cH}{\mathcal{H}}

\newcommand{\cN}{\mathcal{N}}
\newcommand{\cO}{\mathcal{O}}
\newcommand{\cP}{\mathcal{P}}

\newcommand{\cS}{\mathcal{S}}

\newcommand{\cU}{\mathcal{U}}
\newcommand{\cV}{\mathcal{V}}

\newcommand{\cX}{\mathcal{X}}
\newcommand{\cY}{\mathcal{Y}}
\newcommand{\cZ}{\mathcal{Z}}

\newcommand{\bdiv}{\ \mathrm{div} \ }
\newcommand{\Be}{\mathrm{Be}}

\newcommand{\diag}{\mathrm{diag}}
\newcommand{\erfc}{\mathrm{erfc}}

\newcommand{\sym}{\mathrm{sym}}
\newcommand{\tr}{\mathrm{tr}}

\newcommand{\dHam}{d_\mathrm{Ham}}

\newcommand{\CovSampling}{\textsc{CovSamp}}
\newcommand{\NatParamEst}{\textsc{NatParamEst}}

\newcommand{\iprod}[1]{\left\langle #1 \right\rangle}
\newcommand{\paren}[1]{\left( #1 \right)}
\newcommand{\brk}[1]{\left[ #1 \right]}
\newcommand{\brc}[1]{\left\{ #1 \right\}}

\newcommand{\pr}[2]{\underset{#1}{\mathbb{P}}\left[ #2 \right]}
\newcommand{\ex}[2]{\underset{#1}{\mathbb{E}}\left[ #2 \right]}
\newcommand{\var}[2]{\underset{#1}{\mathrm{Var}}\left( #2 \right)}
\newcommand{\llnorm}[1]{\left\| #1 \right\|_2}
\newcommand{\infnorm}[1]{\left\| #1 \right\|_{\infty}}
\newcommand{\fnorm}[1]{\left\| #1 \right\|_F}
\newcommand{\mnorm}[2]{\left\| #2 \right\|_{#1}}

\newcommand{\abs}[1]{\left| #1 \right|}

\newcommand{\eps}{\varepsilon}

\definecolor{shadecolor}{rgb}{0.83, 0.83, 0.83}
\theoremstyle{plain}
\newtheorem{thm}{Theorem}[section]
\newtheorem{theorem}[thm]{Theorem}
\newtheorem{fact}[thm]{Fact}
\newtheorem{proposition}[thm]{Proposition}
\newtheorem{corollary}[thm]{Corollary}
\newtheorem{lemma}[thm]{Lemma}

\theoremstyle{definition}
\newtheorem{definition}[thm]{Definition}

\newtheorem{remark}[thm]{Remark}

\title{New Lower Bounds for Private Estimation and\\a Generalized Fingerprinting Lemma\thanks{Authors are listed in alphabetical order.}}
\author{
    Gautam Kamath\thanks{{\tt g@csail.mit.edu}. Cheriton School of Computer Science, University of Waterloo. Supported by an NSERC Discovery Grant, an unrestricted gift from Google, and a University of Waterloo startup grant.}
\and
    Argyris Mouzakis\thanks{{\tt amouzaki@uwaterloo.ca}. Cheriton School of Computer Science, University of Waterloo. Supported by an NSERC Discovery Grant and a David R. Cheriton Graduate Scholarship.}
\and
    Vikrant Singhal\thanks{{\tt vikrant.singhal@uwaterloo.ca}. Cheriton School of Computer Science, University of Waterloo. Supported by an NSERC Discovery Grant.}
}

\begin{document}

\maketitle

\begin{abstract}

We prove new lower bounds for statistical estimation tasks under the constraint of $\paren{\eps, \delta}$-differential privacy. 
First, we provide tight lower bounds for private covariance estimation of Gaussian distributions.
We show that estimating the covariance matrix in Frobenius norm requires $\Omega\paren{d^2}$ samples, and in spectral norm requires $\Omega\paren{d^{\frac{3}{2}}}$ samples, both matching upper bounds up to logarithmic factors.
The latter bound verifies the existence of a conjectured \emph{statistical gap} between the private and the non-private sample complexities for spectral estimation of Gaussian covariances.
We prove these bounds via our main technical contribution, a broad generalization of the fingerprinting method~\cite{BunUV14} to exponential families.
Additionally, using the private Assouad method of Acharya, Sun, and Zhang~\cite{AcharyaSZ21}, we show a tight $\Omega\paren{\frac{d}{\alpha^2 \eps}}$ lower bound for estimating the mean of a distribution with bounded covariance to $\alpha$-error in $\ell_2$-distance. 
Prior known lower bounds for all these problems were either polynomially weaker or held under the stricter condition of $\paren{\eps,0}$-differential privacy.
\end{abstract}

\newpage

\tableofcontents

\newpage

\section{Introduction}
\label{sec:intro}

The last several years have seen a surge of interest in algorithms for statistical estimation under the constraint of \emph{Differential Privacy} (DP)~\cite{DworkMNS06}.
We now have a rich algorithmic toolbox for private estimation of mean, covariance, and entire distributions, in a variety of settings.

While the community has enjoyed much success designing \emph{algorithms} for estimation tasks, \emph{lower bounds} have been much harder to come by, and consequently, the existing literature lacks a rigorous understanding of some core problems.
Lower bounds against the strong privacy constraint of \emph{pure} DP (i.e., $\paren{\eps, 0}$-DP) are generally not too challenging to prove, most often relying upon the well-known packing technique~\cite{HardtT10}.
Such packing lower bounds are optimal in a broad range of settings (see, e.g.,~\cite{BunKSW19}), and thus are frequently effective at providing tight minimax sample complexity lower bounds.

However, pure DP is a very strong privacy constraint: in most practical circumstances, it suffices to only require the weaker privacy notion of \emph{approximate} DP (i.e., $\paren{\eps,\delta}$-DP)~\cite{DworkKMMN06}, in which case packing lower bounds are no longer applicable.
Proving lower bounds under approximate DP is much more challenging, and accordingly, the state of affairs is less satisfying.
There exist only a couple of techniques which apply in this setting, including the fingerprinting method~\cite{BunUV14} and differentially private analogues of Le Cam's method and Assouad's lemma~\cite{AcharyaSZ18,AcharyaSZ21}.
These techniques can be quite brittle and thus unable to prove lower bounds for some fairly basic settings.
For example, both techniques generally require that the underlying distribution has independent marginals and are thus ineffective at proving tight lower bounds for problems involving correlations such as private covariance estimation.
Furthermore, the fingerprinting method often needs the distribution's marginals to have a precise functional form, such as Gaussian or Bernoulli.
These restrictions make the approaches somewhat brittle for proving lower bounds which differ too much from existing results.
Consequently, we lack approximate DP lower bounds for some fundamental settings which seem to qualitatively differ, including Gaussian covariance estimation and mean estimation of heavy-tailed distributions.

\subsection{Our Results}
\label{subsec:results}

We circumvent these barriers and fill a number of gaps in the literature by providing tight lower bounds for several estimation tasks under the constraint of approximate differential privacy.
Specifically, we address problems including covariance estimation for Gaussians and mean estimation with heavy-tailed data.
Our main technical tool for covariance estimation is a novel generalization of the fingerprinting method to exponential families, while we prove our mean estimation lower bound via DP Assouad's lemma.

Our first result is a lower bound for covariance estimation in Frobenius norm.\footnote{Recall that the Mahalanobis norm of $M$ with respect to $\Sigma$ is $\mnorm{\Sigma}{M} = \fnorm{\Sigma^{-1/2} M \Sigma^{-1/2}}$.
For the lower bound constructions we consider, $\Omega\paren{1} \id \preceq \Sigma \preceq \cO\paren{1} \id$.
In this regime, the Mahalanobis and Frobenius distances are equivalent up to a constant factor.
Thus, we will often use them interchangeably and our lower bound for Frobenius estimation also implies a lower bound for Mahalanobis estimation.
Additionally, the same remark holds for the spectral norms of $M$ and $\Sigma^{-1/2} M \Sigma^{-1/2}$ (which comes up in Theorem~\ref{thm:informal2}).}

\begin{theorem}[Informal version of Theorem~\ref{thm:lb-maha}]
\label{thm:informal1}
For any $\alpha = \cO\paren{1}$, any $\paren{\eps, \delta}$-DP mechanism with $\eps, \delta \in \brk{0, 1}$, and $\delta \le \widetilde{\cO}\paren{\frac{d^2}{n}}$ that takes $n$ samples from an arbitrary $d$-dimensional Gaussian $\cN\paren{0, \Sigma}$ and outputs $M\paren{X}$ satisfying $\ex{X, M}{\fnorm{M\paren{X} - \Sigma}^2} \le \alpha^2$ requires $n \geq \Omega\paren{\frac{d^2}{\alpha \eps}}$.
\end{theorem}

We point out that the above lower bound also implies a lower bound for density estimation of Gaussians with known mean and unknown covariance (see Theorem $1.1$ of~\cite{DevroyeMR18b}).
This nearly matches the $\widetilde{\cO}\paren{d^2}$ upper bound of~\cite{KamathLSU19} up to polylogarithmic factors.\footnote{We note that all our lower bounds in this work are stated in terms of mean squared error, while most of the upper bounds guarantee error at most $\alpha$ with probability $1 - \beta$ for some $\beta > 0$.
In many natural cases (such as when the estimator's range is naturally bounded), lower bounds against MSE can be converted to constant probability statements via a boosting argument (outlined, e.g., in the proof of Theorem 6.1 of~\cite{KamathLSU19}).
The details of such a conversion generally standard when applicable and we do not discuss it further here.}
Previously, $\Omega\paren{d^2}$ lower bounds for Gaussian covariance estimation were only known under the more restrictive constraint of pure differential privacy~\cite{KamathLSU19,BunKSW19}.

Our result for Gaussian covariance estimation in Frobenius norm implies the following result for spectral estimation.

\begin{theorem}[Informal version of Theorem~\ref{thm:lb-spectral}]
\label{thm:informal2}
For any $\alpha = \cO\paren{\frac{1}{\sqrt{d}}}$, any $\paren{\eps, \delta}$-DP mechanism with $\eps, \delta \in \brk{0, 1}$, and $\delta \le \widetilde{\cO}\paren{\frac{d^2}{n}}$ that takes $n$ samples from an arbitrary $d$-dimensional Gaussian $\cN\paren{0, \Sigma}$ and outputs $M\paren{X}$ satisfying $\ex{X, M}{\llnorm{M\paren{X} - \Sigma}^2} \le \alpha^2$ requires $n \geq \Omega\paren{\frac{d^{1.5}}{\alpha \eps}}$.
\end{theorem}

Simply clipping the data at an appropriate radius and adding Gaussian noise to the empirical covariance matrix gives a sample complexity upper bound of $\widetilde \cO\paren{\frac{d^{1.5}}{\alpha\eps}}$ for the same problem, for all $\alpha = \cO\paren{1}$~\cite{DworkTTZ14,KamathLSU19,BiswasDKU20}. 
This implies that our lower bound is tight up to polylogarithmic factors in the regime $\alpha = \cO\paren{\frac{1}{\sqrt{d}}}$.

Note that the sample complexity of the same problem sans privacy constraints is $\Theta\paren{d}$, and thus the additional cost of privacy is polynomial in the dimension $d$.
This is in contrast to Gaussian mean estimation in $\ell_2$-norm and covariance estimation in Frobenius norm, which maintain their non-private sample complexities of $\cO\paren{d}$ and $\cO\paren{d^2}$, respectively.

$\Omega\paren{d^{\frac{3}{2}}}$-sample lower bounds for approximate DP covariance estimation in spectral norm were previously known for \emph{worst-case} data~\cite{DworkTTZ14}.
Our work relaxes this assumption significantly to the Gaussian case while maintaining the same lower bound.
A priori, it was not clear that such a result was even true, as there are frequently gaps between the sample complexity of private estimation for worst-case versus well-behaved distributions for even basic settings, including distributions over the hypercube (see, e.g., Remark 6.4 of~\cite{BunKSW19}).

This lower bound has implications for estimation tasks with scale-dependent error.
For example, consider Gaussian mean estimation in Mahalanobis distance.
A natural way to approach this problem is to first estimate the covariance matrix spectrally, and then estimate the mean after an appropriate rescaling (see, e.g.,~\cite{KamathLSU19}).
Our lower bound shows that any such approach must incur the $\Omega\paren{d^{\frac{3}{2}}}$ cost of spectral estimation, which is greater than the $\cO\paren{d}$ minimax sample complexity of the problem.
Indeed, some recent works~\cite{BrownGSUZ21, LiuKO21} manage to circumvent this roadblock by adopting more direct (but computationally inefficient) methods.

Finally, we prove new lower bounds for mean estimation of distributions with bounded second moments (see Definition~\ref{def:bounded_moments}).

\begin{theorem}[Informal version of Theorem~\ref{thm:lb-second-moment}]
\label{thm:informal3}
For any $\alpha \le 1$, any $\paren{\eps, \delta}$-DP mechanism with $\delta \le \eps$ that takes $n$ samples from an arbitrary distribution over $\bR^d$ with second moments bounded by $1$ and outputs $M\paren{X}$ satisfying $\ex{X, M}{\llnorm{M\paren{X} - \mu}^2} \le \alpha^2$ requires $n \geq \Omega\paren{\frac{d}{\alpha^2 \eps}}$.
\end{theorem}

Similar $\Theta\paren{\frac{d}{\alpha^2 \eps}}$ upper and lower bounds were previously known for distributions satisfying Definition~\ref{def:bounded_moments} for $k = 2$ under the stricter constraint of pure differential privacy~\cite{BarberD14,KamathSU20,HopkinsKM21}.
Also, a similar bound is shown in~\cite{KamathLZ21} for concentrated DP, but under a weaker moment assumption which involves only coordinate-wise projections.
Our result shows that the same bound holds under the weaker constraint of approximate DP for distributions satisfying the stronger moment assumption, and thus no savings can be obtained by relaxing the privacy notion.
In contrast to our other results, we show this lower bound via the DP Assouad method of~\cite{AcharyaSZ21}, thus demonstrating the promise of different approaches for proving lower bounds for differentially private estimation.

\subsection{Our Techniques}
\label{subsec:techniques}

Our main technical contribution is providing a broad generalization of the celebrated fingerprinting method to exponential families.
We say that a distribution with known support (e.g., a subset $S$ of $\bR^d$) belongs to an exponential family if the distribution's density can be written in the form $p_{\eta}\paren{x} = h\paren{x} \exp\paren{\eta^{\top} T\paren{x} - Z\paren{\eta}}$, where $\eta \in \cH \subseteq \bR^k$ (for some $k$) is the \emph{natural parameter vector}, $h \colon S \to \bR_+$ is the \emph{carrier measure}, $T \colon S \to \bR^k$ are the \emph{sufficient statistics}, and $Z \colon \cH \to \bR$ is the \emph{log-partition function}.

Distributions belonging in the same exponential family are parameterized by their \emph{natural parameter vector} $\eta$.
Thus, estimating the natural parameter vector $\eta$ is a well-motivated problem.
Given $n$ samples $X_1, \dots, X_n$ from $p_{\eta}$, it suffices to know the values of the sufficient statistics $T\paren{X_1}, \dots, T\paren{X_n}$ instead of the samples themselves in order to estimate $\eta$ (see, e.g.,~\cite{Jordan10}).
We denote the mean and covariance of $T$ by $\mu_T$ and $\Sigma_T$, respectively.

Our extension of the fingerprinting lemma for exponential families crucially relies on the above observation.
Traditional fingerprinting proofs (see, e.g.,~\cite{BunUV14, SteinkeU17a, DworkSSUV15, BunSU17, SteinkeU17b, KamathLSU19}) consider the trade-off between the accuracy of a mechanism and the correlation of its output with the input samples.
Conversely, our lemma involves reasoning about the correlation of the output with the sufficient statistics of the samples.

Another important aspect of our technique is that it allows us to prove lower bounds for non-product distributions.
Having considered the trade-off between accuracy and correlation with the input, the second step in fingerprinting proofs is to upper-bound the correlation terms using the definition of differential privacy.
All prior proofs however managed to do this only for product distributions, significantly limiting their applicability.
On the other hand, our work leverages properties of exponential families, allowing us to express the upper bound in terms of the covariance matrix $\Sigma_T$.

The above observations lead to the following theorem, giving lower bounds for private estimation for exponential families.

\begin{theorem}[Informal version of Theorem~\ref{thm:lower_bound}]
\label{thm:informal4}
Let $X \sim p_{\eta}^{\bigotimes n}$ be a dataset, where $p_{\eta}$ belongs to an exponential family $\cE\paren{T, h}$ with natural parameter vector $\eta \in \mathbb{R}^k$ randomly sampled from an axis-aligned hyper-rectangle with side lengths $R \in \mathbb{R}^k$ and midpoint $m \in \mathbb{R}^k$.
Let $X_{\sim i}$ denote the dataset where $X_i$ has been replaced with a fresh sample $X_i' \sim p_\eta$.
Let $M$ be an $\paren{\eps, \delta}$-DP mechanism with accuracy guarantee $\ex{X, M}{\llnorm{M\paren{X} - \paren{\eta - m}}^2} \le \frac{\llnorm{R}^2}{24}$.
Then, for any $T > 0$, it holds that:
\[
    n \brc{2 \delta T + 2 \eps \ex{\eta}{\sqrt{\ex{X_{\sim i}, M}{s^{\top} \Sigma_T s}}} + 2 \ex{\eta}{\int\limits_T^{\infty} \pr{X_i}{\llnorm{T\paren{X_i} - \mu_T} > \frac{4 t}{\infnorm{R}^3 \sqrt{k}}} \, dt}} \geq \frac{\llnorm{R}^2}{24},
\]
where $s \in \bR^k$ is a vector measuring a proxy for the coordinate-wise correlation between the true parameter vector and the mechanism's error when an arbitrary sample is redrawn.
\end{theorem}

In the above, $M$ does not estimate $\eta$ itself, but rather its deviation from $m$.
The two estimation problems are equivalent and we use this formulation for purely technical reasons.
Additionally, observe that the sufficient statistics of the underlying distribution appear both in the second and third terms of the LHS.
This should be contrasted with older fingerprinting proofs (see, e.g., the proofs of Lemmas $6.2$ and $6.7$ from~\cite{KamathLSU19}), which instead involve the variance and the tail probabilities of the underlying distribution, respectively.
This difference is because, as we described previously, the correlation is measured with respect to the sufficient statistics of the samples, instead of the samples themselves.

It is a known fact that Gaussian distributions $\cN\paren{0, \Sigma}$ are an exponential family.
This allows us to apply our fingerprinting lemma for exponential families to get the lower bound for estimation with respect to the Frobenius norm.
On the other hand, our result for spectral estimation is the consequence of a reduction from Frobenius estimation to spectral estimation which makes use of the standard inequality $\fnorm{A} \le \sqrt{d} \llnorm{A}$ (hence, the constraint $\alpha \le \cO\paren{\frac{1}{\sqrt{d}}}$).

Finally, we remark that our lemma captures previously studied settings for approximate DP lower bounds in statistical estimation, including mean estimation for binary product distributions and identity-covariance Gaussians (see~\cite{KamathLSU19}).

\subsection{Related Work}
\label{subsec:related_work}

The fingerprinting technique for proving lower bounds was introduced by Bun, Ullman, and Vadhan~\cite{BunUV14}, which relied on the existence of \emph{fingerprinting codes}~\cite{BonehS95, Tardos03}.
Since then, the technique has been significantly simplified and refined in a number of ways~\cite{SteinkeU17a, DworkSSUV15, BunSU17, SteinkeU17b, KamathLSU19,CaiWZ19}, including the removal of fingerprinting codes and distilling the main technical component into the \emph{fingerprinting lemma}.
Beyond the aforementioned settings of mean estimation of Gaussians and distributions over the hypercube~\cite{BunUV14,SteinkeU15,DworkSSUV15, KamathLSU19}, fingerprinting lower bounds have also been applied in settings including private empirical risk minimization~\cite{BassilyST14} and private spectral estimation~\cite{DworkTTZ14}.
Differentially private analogues of Fano, Le Cam, and Assouad were first considered in the local DP setting~\cite{DuchiJW13}, and more recently under the constraint of central DP~\cite{AcharyaSZ18, AcharyaSZ21}.

On the upper bounds side, the most relevant body of work is that which focuses on mean and covariance estimation for distributions satisfying a moment bound, see, e.g.,~\cite{BarberD14, KarwaV18, BunKSW19, BunS19, KamathLSU19, CaiWZ19, KamathSU20, WangXDX20, DuFMBG20, BiswasDKU20, AdenAliAK21, BrownGSUZ21, KamathLZ21, HuangLY21, KamathMSSU21, AshtianiL21}.
Our lower bounds nearly match these upper bounds for Gaussian covariance estimation~\cite{KamathLSU19} and heavy-tailed mean estimation~\cite{KamathSU20,HopkinsKM21}.
Previously, these bounds were conjectured to be optimal, albeit on somewhat shaky grounds including evidence from lower bounds proven under pure DP or non-rigorous connections with lower bounds in related settings.
Our results rigorously prove optimality and confirm these conjectures.
Other works study private statistical estimation in different and more general settings, including mixtures of Gaussians~\cite{KamathSSU19, AdenAliAL21}, graphical models~\cite{ZhangKKW20}, discrete distributions~\cite{DiakonikolasHS15}, and median estimation~\cite{AvellaMedinaB19, TzamosVZ20}.
Some recent directions involve guaranteeing user-level privacy~\cite{LiuSYKR20, LevySAKKMS21}, or a combination of local and central DP for different users~\cite{AventDK20}.
See~\cite{KamathU20} for further coverage of DP statistical estimation.

\subsection{Discussion of Independent Work}

Following the initial publication of our results, we became aware of the independent works of Cai, Wang, and Zhang~\cite{CaiWZ20,CaiWZ23}.\footnote{The former work is prior to ours, and the latter is subsequent, but appears to be an updated and expanded version of the former.}
These works also present a generalization of fingerprinting (which they call the \emph{score attack}) that also goes beyond prior works focusing on Gaussians and binary product distributions.
Their result is framed in terms of general parametric models that satisfy a number of regularity conditions (which exponential families satisfy).
Conversely, our work focuses exclusively on exponential families, but our ideas can directly be generalized to their setting by observing that: (i) $T\paren{X} - \mu_T$ coincides with the \emph{score statistic} when evaluated over an exponential family, and (ii) the covariance matrix of the sufficient statistics $\Sigma_T$ is the \emph{Fisher Information Matrix} for exponential families $\cE\paren{T, h}$.
Thus, the core lemmas in both papers are \emph{morally equivalent}.

The main difference between the two results has to do with the specific settings investigated, and the constructions analyzed.
In particular, their result considers prior distributions that may be supported either over a finite or an infinite domain.
To handle this case, they assume that the prior involved in the application of the method is continuously differentiable and apply Stein's Lemma.
On the other hand, we work with uniform priors over finite axis-aligned hyper-rectangles that are subsets of the parameter space (which is one of the cases they consider) and multiply our correlation terms by appropriate scaling factors to deal with the issue that the density does not vanish as we approach the boundary of the hyper-rectangles from their interior.
Additionally, in order to handle this, we assume that our estimators are not estimating the parameter itself, but rather its deviation from the midpoint of the aforementioned hyper-rectangle.
Another (minor) difference comes up when observing the upper bound due to privacy for the correlation terms (referred in~\cite{CaiWZ23} as the \emph{soundness} case).
In~\cite{CaiWZ23}, the upper bound is written in terms of the maximum eigenvalue of the Fisher Information Matrix, whereas our bound uses a quadratic form involving $\Sigma_T$.
Our choice not to upper-bound the corresponding term via the maximum eigenvalue was conscious because it facilitates the application of the result in the context considered in the paper.

In terms of applications, we apply our method to Gaussian covariance estimation, whereas they consider Generalized Linear Models, the Bradley-Terry-Luce ranking model, and non-parametric regression.
We comment that, the way we have implemented the method, we put special care on ensuring that our $\paren{\eps, \delta}$-DP lower bounds apply for $\delta$ as large as $\widetilde{\cO}\paren{\frac{1}{n}}$ (which is essentially optimal), while some of the bounds in~\cite{CaiWZ23} have not been optimized with respect to this aspect and hold for $\delta \le \cO\paren{\frac{1}{n^{1 + \gamma}}}$, for some $\gamma > 0$.

\section{Preliminaries}
\label{sec:prelim}

\textbf{General Notation.}
We use the notations: $\brk{n} \coloneqq \brc{1, 2, \dots, n}, \brk{a \pm R} \coloneqq \brk{a - R, a + R}$, and $\bR_+ \coloneqq \left[0, \infty\right)$.
Given any distribution $\cD$, $\cD^{\bigotimes n}$ denotes the \emph{product measure} where each \emph{marginal distribution} is $\cD$.
Thus, if we are given $n$ independent samples from $\cD$, we write $\paren{X_1, \dots, X_n} \sim \cD^{\bigotimes n}$.
Also, depending on the context, we may use capital Latin characters like $X$ to denote either an individual sample from a distribution or a collection of samples $X \coloneqq \paren{X_1, \dots, X_n}$.
To denote the $j$-th component of a vector, we will use either a subscript (e.g., $X_j$, if the vector is $X$) or a superscript (e.g., $X_i^j$, if the vector is $X_i$).
The convention used will be clear depending on the context.
The \emph{unit $\ell_2$-sphere} in $d$-dimensions that is centered at the origin is denoted by $\bS^{d - 1}$.
Finally, we give the following definition for distributions with bounded $k$-th moments:

\begin{definition}
\label{def:bounded_moments}
Let $\cD$ be a distribution over $\bR^d$.
We say that $\cD$ has $k$-th moments bounded by $C$ if $\ex{X \sim \cD}{\abs{\iprod{X - \ex{X \sim \cD}{X}, v}}^k} \le C, \forall v \in \bS^{d - 1}$.
\end{definition}

\medskip

\textbf{Linear Algebra Preliminaries.}
For a vector $v \in \bR^d$, $v_i$ refers to its $i$-th component and $v_{- i} \in \bR^{d - 1}$ describes the vector one gets by removing the $i$-th component from $v$.
We use $\Vec{1}_d$ to refer to the all $1$s vector in $\bR^d$.
For a pair of vectors $x, y \in \bR^d$, we write $x \le y$ if and only if $x_i \le y_i, \forall i \in \brk{n}$.
Given a matrix $M \in \bR^{d \times d}$, we denote the element at the $i$-th row and the $j$-th column by $M_{i j}$ or $\paren{M}_{i j}$.
Also, we denote the \emph{canonical flattening} operation of $M$ into a vector by $v = M^{\flat} \in \bR^{d^2}$ and the inverse operation by $M = v^{\#} \in \bR^{d \times d}$.
We denote its \emph{spectral norm} by $\llnorm{M} \coloneqq \sup\limits_{v \in \bS^{d - 1}} \llnorm{M v}$.
For symmetric matrices $M$, this reduces to $\llnorm{M} = \sup\limits_{v \in \bS^{d - 1}} \abs{v^{\top} M v}$, which is known to be equal to the largest eigenvalue of $M$ (in absolute value).
Also, we denote the \emph{trace} of $M$ by $\tr\paren{M} \coloneqq \sum\limits_{i \in \brk{d}} M_{i i}$, and its \emph{Frobenius norm} by $\fnorm{M} \coloneqq \sqrt{\tr\paren{M^{\top} M}} = \sqrt{\sum\limits_{i = 1}^d \sum\limits_{j = 1}^d M_{i j}^2} = \llnorm{M^{\flat}}$.
For symmetric matrices, the latter is equal to the square root of the sum of the squares of its eigenvalues.
Moreover, we consider the \emph{Mahalanobis norm} (with respect to a non-degenerate covariance matrix $\Sigma$), both for vectors $\paren{\mnorm{\Sigma}{x} \coloneqq \llnorm{\Sigma^{- \frac{1}{2}} x}}$ and matrices $\paren{\mnorm{\Sigma}{M} \coloneqq \fnorm{\Sigma^{- \frac{1}{2}} M \Sigma^{- \frac{1}{2}}}}$.
Furthermore, we use the notation $\iprod{\cdot, \cdot}$ to denote \emph{inner products}, either between vectors $x, y \in \bR^d \ \paren{\iprod{x, y} \coloneqq x^{\top} y}$ or between matrices $A, B \in \bR^{d \times d} \ \paren{\iprod{A, B} \coloneqq \tr\paren{A B^{\top}}}$.
The inner product of a pair of $d \times d$ matrices $A, B$ can be reduced to an inner product between vectors in $\bR^{d^2}$ via the identity $\iprod{A, B} = \iprod{A^{\flat}, B^{\flat}}$.
Additionally, given a pair of symmetric matrices $A, B \in \bR^d$, we write $A \succeq B$ if and only if $x^{\top} \paren{A - B} x \geq 0, \forall x \in \bR^d$.
Finally, given matrices $A \in \bR^{n \times m}$ and $B \in \bR^{k \times \ell}$, the \emph{Kronecker product} of $A$ with $B$ is denoted by $A \otimes B \in \bR^{\paren{n k} \times \paren{m \ell}}$.
More linear algebra facts appear in Appendix~\ref{sec:facts_lin_alg}.

\subsection{Privacy Preliminaries}
\label{subsec:privacy_prelim}

We define differential privacy and state its closure under post-processing property.

\begin{definition}[Differential Privacy (DP)~\cite{DworkMNS06}]
\label{def:dp}
A randomized algorithm $M \colon \cX^n \rightarrow \cY$ satisfies $\paren{\eps, \delta}$-differential privacy ($\paren{\eps, \delta}$-DP) if for every pair of neighboring datasets $X, X' \in \cX^n$ (i.e., datasets that differ in exactly one entry), we have:
\[
    \pr{M}{M\paren{X} \in Y} \le e^{\eps} \cdot \pr{M}{M\paren{X'} \in Y} + \delta, \forall Y \subseteq \cY.
\]
When $\delta = 0$, we say that $M$ satisfies $\eps$-differential privacy or pure differential privacy.
\end{definition}

\begin{lemma}[Post Processing~\cite{DworkMNS06}]
\label{lem:post_processing}
If $M \colon \cX^n \rightarrow \cY$ is $\paren{\eps, \delta}$-DP, and $P \colon \cY \to \cZ$ is any randomized function, then the algorithm $P \circ M$ is $\paren{\eps,\delta}$-DP.
\end{lemma}

Finally, we discuss a recent technique from \cite{AcharyaSZ21} for proving lower bounds for estimation problems under approximate DP.
It is a private version of the well-known Assouad Lemma~\cite{Yu97}.
Let $\cP$ be a family of distributions over $\cX^n$, where $\cX$ is the data universe and $n$ is the number of samples.
Let $\theta \colon \cP \to \Theta$ be the quantity associated with the distribution that we want to estimate, and let $\ell \colon \Theta \times \Theta \to \bR_+$ be the pseudo-metric (loss function) for estimating $\theta$.
We define the risk of an estimator $\hat{\theta} \colon \cX^n \to \Theta$ as $\max\limits_{p \in \cP} \ex{X \sim p}{\ell\paren{\hat{\theta}\paren{X},\theta\paren{p}}}$.
The minimax risk of $\paren{\eps, \delta}$-DP estimators for a statistical task is defined as:
\[
    R\paren{\cP, \ell, \eps, \delta} \coloneqq \min\limits_{\hat{\theta} \text{ is } \paren{\eps,\delta} \text{-DP}} \max\limits_{p \in \cP}\ex{X \sim p}{\ell\paren{\hat{\theta}\paren{X}, \theta\paren{p}}}.
\]
Let $\cV \subseteq \cP$ be a subset of distributions indexed by the points in the hypercube $\cE_d = \brc{\pm 1}^d$.
Suppose there exists a $\tau \in \bR$, such that for all $u, v \in \cE_d$, $\ell\paren{\theta\paren{p_u}, \theta\paren{p_v}} \geq 2 \tau \cdot \sum\limits_{i = 1}^d \mathds{1}\brc{u_i \neq v_i}$.
For every coordinate $i \in [d]$, define the following mixture distributions:
\[
    p_{+ i} \coloneqq \frac{2}{\abs{\cE_d}} \sum\limits_{e \in \cE_d \colon e_i = 1}{p_e}
    \text{   and   }
    p_{- i} \coloneqq \frac{2}{\abs{\cE_d}} \sum\limits_{e \in \cE_d \colon e_i = -1}{p_e}.
\]
The DP Assouad's method provides a lower bound for the minimax risk by considering the problem of distinguishing between $p_{+ i}$ and $p_{- i}$.

\begin{lemma}[DP Assouad's Method \cite{AcharyaSZ21}]
\label{lem:dp-assouad}
For all $i \in [d]$, let $\phi_i \colon \cX^n \to \brc{\pm 1}$ be a binary classifier.
Then we have the following:
\[
    R\paren{\cP, \ell, \eps, \delta} \geq \frac{\tau}{2} \cdot \sum\limits_{i=1}^{d} \min\limits_{\phi_i \text{ is } \paren{\eps, \delta} \text{-DP}} \paren{\pr{X \sim p_{+ i}}{\phi_i\paren{X} \neq 1} + \pr{X \sim p_{- i}}{\phi_i\paren{X} \neq -1}}.
\]
Moreover, if $\forall i \in \brk{d}$, there exists a coupling $\paren{X, Y}$ between $p_{+ i}$ and $p_{- i}$ with $\ex{}{\dHam\paren{X, Y}} \le D$, then:
\[
    R\paren{\cP, \ell, \eps, \delta} \geq \frac{d \tau}{2} \cdot \paren{0.9 e^{-10 \eps D} - 10 \delta D}.
\]
\end{lemma}

\subsection{Exponential Families}
\label{subsec:exp_fams}

In this work, we make frequent use of exponential families due to their expressive power.
This gives us a unified way of dealing with multiple commonly used distributions (e.g., Gaussians, Bernoullis).
We first state the definition of exponential families.

\begin{definition}[Exponential Family]
\label{def:exp_fams}
Let $S \subseteq \bR^d$ and $h \colon S \to \bR_+, T \colon S \to \bR^k$.
Given $\eta \in \bR^k$, we say that the density function $p_{\eta}$ belongs to the \emph{exponential family} $\cE\paren{T, h}$ if its density can be written in the form:
\[
    p_{\eta}\paren{x} = h\paren{x} \exp\paren{\eta^{\top} T\paren{x} - Z\paren{\eta}}, \forall x \in S,
\]
where:
\[
    Z\paren{\eta} = \ln\paren{\int\limits_S h\paren{x} \exp\paren{\eta^{\top} T\paren{x}} \, dx}.
\]
The functions $h$ and $T$ are referred to as the \emph{carrier measure} and the \emph{sufficient statistics} of the family, respectively.
Additionally, $Z$ is known as the \emph{log-partition function} and $\eta$ is the \emph{natural parameter vector}.
Finally, we denote the \emph{range of natural parameters} by $\cH \subseteq \bR^k$, which is the set of values of $\eta$ for which the log-partition function is well-defined $\paren{Z\paren{\eta} < \infty}$.\footnote{The definition is phrased in terms of continuous distributions, but applies to discrete distributions as well.
Indeed, $p_{\eta}$ will be a \emph{probability mass function}, whereas the integral will be replaced by a sum in the definition of $Z$.
We gave this version of the definition to avoid measure-theoretic notation, since it does not add much to our results.}
\end{definition}

We appeal to the following properties of exponential families.

\begin{proposition}[see~\cite{Jordan10}]
\label{prop:exp_prop}
Let $\cE\paren{T, h}$ be an exponential family parameterized by $\eta \in \bR^k$.
Then, for all $\eta$ in the interior of $\cH$, the following hold:

\begin{enumerate}
	\item The mean of the sufficient statistics is:
	\[
	    \mu_T \coloneqq \ex{X \sim p_{\eta}}{T\paren{X}} = \nabla_{\eta} Z\paren{\eta}.
	\]
	\item The covariance matrix of the sufficient statistics is:
	\[
	    \Sigma_T \coloneqq \ex{X \sim p_{\eta}}{\paren{T\paren{X} - \mu_T} \paren{T\paren{X} - \mu_T}^{\top}} = \nabla^2_{\eta} Z\paren{\eta}.
	\]
	\item For all $s \in \bR^k$ such that $\eta + s \in \cH$, the MGF of the sufficient statistics is:
	\[
	    \ex{X \sim p_{\eta}}{\exp\paren{s^{\top} T\paren{X}}} = \exp\paren{Z\paren{\eta + s} - Z\paren{\eta}}.
	\]
	\item The natural parameter range $\cH$ is a convex set, and the log-partition function $Z$ is convex.
\end{enumerate}
\end{proposition}

For a number of facts from probability and statistics (both related and unrelated to exponential families), we refer the reader to Appendix~\ref{sec:facts_prob_stats}.

\section{Fingerprinting Proofs for Exponential Families}
\label{sec:fing_exp_fams}

This section is split into three parts.
Section~\ref{subsec:intuition} provides an overview of existing fingerprinting proofs and motivation for our result.
Also, it introduces notation we will use in subsequent sections.
Then, in Section~\ref{subsec:fing_lem}, we prove our new fingerprinting lemma, and in Section~\ref{subsec:LB}, we provide a general recipe about how the lemma can be applied.
We show that our technique can be used to recover existing lower bounds in Appendix~\ref{sec:existing},\footnote{The bounds we recover in Appendix~\ref{sec:existing} are all from~\cite{KamathLSU19}.
These include one lower bound for mean estimation of binary product distributions, and one for Gaussian mean estimation.
In the process of proving the latter, we identified a bug in the original proof from~\cite{KamathLSU19}, which we correct in this manuscript.
This resulted in a slightly different range for $\delta$ than the one given in the original theorem and proof.}
and give new lower bounds for covariance estimation in Section~\ref{sec:lb_cov_est}.
All the results that we give here hold for any exponential family (see Definition~\ref{def:exp_fams}).
Thus, we will assume that the data-generating distribution is some $p_{\eta}$ over $S \subseteq \bR^d$ with natural parameter vector $\eta \in \cH \subseteq \bR^k$ that belongs to a family $\cE\paren{T, h}$.

\subsection{Overview of Fingerprinting Proofs and Intuition for Our Result}
\label{subsec:intuition}

Fingerprinting proofs capture the trade-off between the accuracy of an estimator and the correlation between its output and the input samples.
For simplicity, assume that the data-generating distribution is univariate and comes from a family that is parameterized by its mean $\mu$, which we want to estimate.
Standard techniques based on minimax theory reduce testing to estimation and lead to lower bounds for the latter by considering a finite family of distributions for which testing is hard.
In contrast, fingerprinting lower bounds assume we have a subset of the domain of $\mu$ (commonly a $0$-centered real interval $\brk{\pm R}$) and a uniform prior over it, from which a value of $\mu \sim \cU\brk{\pm R}$ is drawn.
Thus, we randomly pick one of the distributions $\cD_{\mu}$ in the family and draw $n$ independent samples $X \sim \cD_{\mu}^{\otimes n}$.
Then, given an estimator $f \colon \bR^n \to \brk{\pm R}$ for $\mu$, we consider the quantity:
\[
    Z_i \coloneqq c\paren{R, \mu} \paren{f\paren{X} - \mu} \paren{X_i - \mu}, \forall \in \brk{n},
\]
which measures the correlation between the output of $f$ and the $i$-th sample ($c\paren{R, \mu} > 0$ is a scaling factor depending on the distribution family considered and exists for technical purposes).

Defining $Z \coloneqq \sum\limits_{i \in \brk{n}} Z_i$, all proofs first show a lower bound $L$ on:
\[
    \ex{\mu, X}{\underbrace{Z}_{\text{correlation}} + \underbrace{\paren{f\paren{X} - \mu}^2}_{\text{accuracy}}},
\]
implying that our estimator cannot at the same time be accurate (with respect to the Mean-Squared-Error (MSE)) and exhibit low correlation with all the individual samples.
This component of the proof is traditionally called the \emph{fingerprinting lemma}.

The steps we have described so far do not involve any privacy assumptions.
Privacy comes in the second step of the proof, where we upper-bound the correlation terms $\ex{\mu, X}{Z_i}$.
Specifically, we assume that $f$ is a private mechanism with MSE at most $\frac{L}{2}$.
Due to privacy, all correlation terms $Z_i$ should be upper-bounded by the same $U = U\paren{\eps, \delta}$.
Applying this for each sample leads to an inequality of the form $n U \geq \frac{L}{2} \iff n \geq \frac{L}{2 U}$, so the only thing that remains is to identify the appropriate parameter range for $\delta$ to prove the desired bound.

Our work is motivated by a fundamental weakness of existing fingerprinting proofs.
To prove lower bounds for high-dimensional distributions via fingerprinting, past works have assumed that the underlying distribution is a product distribution.
This made it possible to work component-wise using techniques for single-dimensional distributions.
This assumption is somewhat restrictive and makes it difficult to generalize the approach to settings with a richer correlation structure.
We prove a general lemma for exponential families $\cE\paren{T, h}$ and, in the process, show how to go beyond product distributions.

The proof necessitates us to shift our focus away from mean estimation.
Indeed, there is no general closed-form expression for the mean of a distribution belonging to an exponential family, and it is generally uncommon for such families to be parameterized by it.
Instead, we consider estimating the natural parameter vector $\eta$.
We assume the existence of vectors $\eta^{\paren{1}}, \eta^{\paren{2}}$ with $\eta^{\paren{1}} \le \eta^{\paren{2}}$, and define the intervals $I_j \coloneqq \brk{\eta^{\paren{1}}_j, \eta^{\paren{2}}_j}$ for each coordinate $j \in \brk{k}$.
Then, assuming that $\bigotimes\limits_{j \in \brk{k}} I_j \subseteq \cH$, we draw a natural parameter vector $\eta$ such that each coordinate $\eta_j$ is sampled uniformly at random from $I_j$.
Observe that $I_j$ may not be centered around the origin, which poses a minor technical difficulty.
To deal with this, we assume that our estimator does not estimate $\eta$ itself, but rather the deviation of $\eta$ from the midpoint $m \coloneqq \frac{\eta^{\paren{1}} + \eta^{\paren{2}}}{2}$ of $\bigotimes\limits_{j \in \brk{k}} I_j$.
The two problems are equivalent because any estimator $f \colon S^n \to \bigotimes\limits_{j \in \brk{k}} \brk{\pm \frac{\eta^{\paren{2}}_j - \eta^{\paren{1}}_j}{2}}$ for the latter can be transformed to an estimator for the former, simply by adding $m$ to it (and vice-versa).
Defining $R \coloneqq \eta^{\paren{2}} - \eta^{\paren{1}}$, the correlation between the estimate for $\eta_j$ and the sample $X_i$ becomes:
\begin{align}
    Z_i^j \coloneqq \brk{\frac{R_j^2}{4} - \paren{\eta_j - m_j}^2} \brk{f_j\paren{X} - \paren{\eta_j - m_j}} \paren{T_j\paren{X_i} - \mu_{T, j}}, \forall i \in \brk{n}, j \in \brk{k}. \label{eq:corr_term1}
\end{align}
Observe that now, $Z_i^j$ measures the correlation between $f$ and the \emph{sufficient statistics} of the samples.
This is in contrast to older fingerprinting proofs, in which the $Z_i^j$ measure correlation between $f$ and the samples themselves.
We consider our alternate parameterization to be natural in the context of exponential families.
Indeed, we know that, if we want to estimate the natural parameter vector of a distribution that belongs to an exponential family given a dataset $X \coloneqq \paren{X_1, \dots, X_n}$, it suffices to know the values $T\paren{X_1}, \dots, T\paren{X_n}$ instead of the samples themselves (see~\cite{Jordan10}).
Thus, despite $f$ receiving $X$, correlation ends up being measured with respect to $T\paren{X_i}, i \in \brk{n}$.

In addition to the above, we introduce three more correlation terms:
\begin{align}
    Z_i &\coloneqq \sum\limits_{j \in \brk{k}} Z_i^j = \sum\limits_{j \in \brk{k}} \brk{\frac{R_j^2}{4} - \paren{\eta_j - m_j}^2} \brk{f_j\paren{X} - \paren{\eta_j - m_j}} \paren{T_j\paren{X_i} - \mu_{T, j}}, \label{eq:corr_term2} \\
    Z^j &\coloneqq \sum\limits_{i \in \brk{n}} Z_i^j = \brk{\frac{R_j^2}{4} - \paren{\eta_j - m_j}^2} \brk{f_j\paren{X} - \paren{\eta_j - m_j}} \sum\limits_{i \in \brk{n}} \paren{T_j\paren{X_i} - \mu_{T, j}}, \label{eq:corr_term3} \\
      Z &\coloneqq \sum\limits_{i \in \brk{n}} Z_i = \sum\limits_{j \in \brk{k}} Z^j. \label{eq:corr_term4}
\end{align}
We motivate each of the above.
(\ref{eq:corr_term2}) describes the correlation between the $i$-th sample and our estimates for all components of $\eta$.
Moreover, (\ref{eq:corr_term3}) describes the correlation between our estimate for $\eta_j$ and all of the samples.
Finally, (\ref{eq:corr_term4}) describes the total correlation between the output of $f$ and the dataset.
Our fingerprinting lemma for exponential families will involve upper and lower-bounding the quantity:
\[
    \ex{\eta, X}{Z + \llnorm{f\paren{X} - \paren{\eta - m}}^2}.
\]

\subsection{The Main Lemma}
\label{subsec:fing_lem}

We prove our new fingerprinting lemma in two parts.
First, we assume we have a distribution $p_{\eta}$ that belongs to an exponential family and we want to estimate $\eta_j$ whereas $\eta_{- j}$ is fixed (Lemma~\ref{lem:fing}).
This setting is similar to the ones considered in~\cite{BunSU17, KamathLSU19}, which feature similar lemmas.
We then leverage Lemma~\ref{lem:fing} to get a result when the goal is to estimate all of $\eta$ (Lemma~\ref{lem:fing_gen}).

\begin{lemma}
\label{lem:fing}
Let $p_{\eta}$ be a distribution over $S \subseteq \bR^d$ belonging to an exponential family $\cE\paren{T, h}$ with natural parameter vector $\eta \in \cH \subseteq \bR^k$.
We assume that $I_j$ is non-degenerate, and that $\eta_j$ is randomly generated by drawing $\eta_j \sim \cU\paren{I_j}$, whereas $\eta_{- j}$ is fixed.
Then, for any function $f_j \colon S^n \to \brk{\pm \frac{R_j}{2}}$ that takes as input a dataset $X \coloneqq \paren{X_1, \dots, X_n} \sim p_{\eta}^{\bigotimes n}$, we have:
\[
    \ex{\eta_j, X}{Z^j + \brk{f_j\paren{X} - \paren{\eta_j - m_j}}^2} \geq \frac{R_j^2}{12}.
\]
\end{lemma}

\begin{proof}
Let $g_j \colon \bigotimes\limits_{j \in \brk{k}} I_j \to \brk{\pm \frac{R_j}{2}}$ with $g_j\paren{\eta} \coloneqq \ex{X}{f_j\paren{X}}$.
$g_j$ can equivalently be written as:
\[
    g_j\paren{\eta} = \bigintsss\limits_{\bR^d} \dots \bigintsss\limits_{\bR^d} \paren{\prod\limits_{i \in \brk{n}} p_{\eta}\paren{x_i}} f_j\paren{x_1, \dots, x_n} \, dx_1 \dots dx_n.
\]
Observe that:
\begin{align*}
    \paren{\prod\limits_{i \in \brk{n}} p_{\eta}\paren{x_i}} &= \paren{\prod\limits_{i \in \brk{n}} h\paren{x_i}} \exp{\paren{\iprod{\eta, \sum\limits_{i \in \brk{n}} T\paren{x_i}} - n Z\paren{\eta}}} \\
    \implies \frac{\partial}{\partial \eta_j} \paren{\prod\limits_{i \in \brk{n}} p_{\eta}\paren{x_i}} &= \paren{\prod\limits_{i \in \brk{n}} p_{\eta}\paren{x_i}} \sum\limits_{i \in \brk{n}} \paren{T_j\paren{x_i} - \frac{\partial}{\partial \eta_j} \paren{Z\paren{\eta}}} \\
                                                                                    &= \paren{\prod\limits_{i \in \brk{n}} p_{\eta}\paren{x_i}}  \sum\limits_{i \in \brk{n}} \paren{T_j\paren{x_i} - \mu_{T, j}},
\end{align*}
where the last equality uses the first part of Proposition~\ref{prop:exp_prop}.

Calculating the partial derivative of $g_j$ with respect to $\eta_j$ and changing the order of differentiation and integration, we get:
\begin{align*}
    \frac{\partial}{\partial \eta_j} \paren{g_j\paren{\eta}} &= \bigintsss\limits_{\bR^d} \dots \bigintsss\limits_{\bR^d} \paren{\prod\limits_{i \in \brk{n}} p_{\eta}\paren{x_i}} f_j\paren{x_1, \dots, x_n} \sum\limits_{i \in \brk{n}} \paren{T_j\paren{x_i} - \mu_{T, j}} \, dx_1 \dots dx_n \\
                                                             &= \ex{X}{f_j\paren{X} \sum\limits_{i \in \brk{n}} \paren{T_j\paren{X_i} - \mu_{T, j}}} \\
                                          \overset{\paren{a}}&{=} \ex{X}{f_j\paren{X} \sum\limits_{i \in \brk{n}} \paren{T_j\paren{X_i} - \mu_{T, j}}} - \ex{X}{\paren{\eta_j - m_j} \sum\limits_{i \in \brk{n}} \paren{T_j\paren{X_i} - \mu_{T, j}}} \\
                                                             &= \ex{X}{\brk{f_j\paren{X} - \paren{\eta_j - m_j}} \sum\limits_{i \in \brk{n}} \paren{T_j\paren{X_i} - \mu_{T, j}}},
\end{align*}
where $\paren{a}$ relies on the observation that, conditioned on $\eta_j$, $\eta_j - m_j$ and $\sum\limits_{i \in \brk{n}} \paren{T_j\paren{X_i} - \mu_{T, j}}$ are independent and $\ex{X}{\sum\limits_{i \in \brk{n}} \paren{T_j\paren{X_i} - \mu_{T, j}}} = 0$.

We proceed to define the function $G_j \colon \bigotimes\limits_{j \in \brk{k}} I_j \to \bR$ as:
\begin{align*}
    G_j\paren{\eta} \coloneqq \ex{X}{Z^j} &= \brk{\frac{R_j^2}{4} - \paren{\eta_j - m_j}^2} \ex{X}{\brk{f_j\paren{X} - \paren{\eta_j - m_j}} \sum\limits_{i \in \brk{n}} \paren{T_j\paren{X_i} - \mu_{T, j}}} \\
                                          &= \brk{\frac{R_j^2}{4} - \paren{\eta_j - m_j}^2} \frac{\partial g_j}{\partial \eta_j}\paren{\eta}.
\end{align*}
Calculating the expectation of $G_j$ with respect to $\eta_j$ yields:
\begin{align*}
    \ex{\eta_j}{G_j\paren{\eta}} &= \ex{\eta_j}{\brk{\frac{R_j^2}{4} - \paren{\eta_j - m_j}^2} \frac{\partial g_j}{\partial \eta_j}\paren{\eta}} \\
                                 &= \frac{1}{\abs{I_j}} \int\limits_{\eta^{\paren{1}}_j}^{\eta^{\paren{2}}_j} \brk{\frac{R_j^2}{4} - \paren{\eta_j - m_j}^2} \frac{\partial g_j}{\partial \eta_j}\paren{\eta_1, \dots, \eta_j, \dots, \eta_k} \, d\eta_j \\
              \overset{\paren{a}}&{=} \cancelto{0}{\frac{1}{\abs{I_j}} \brk{\frac{R_j^2}{4} - \paren{\eta_j - m_j}^2} g_j\paren{\eta_1, \dots, \eta_j, \dots, \eta_k} \Bigg|_{\eta_j = \eta^{\paren{1}}_j}^{\eta_j = \eta^{\paren{2}}_j}} \\ 
                                 &\quad+ 2 \frac{1}{\abs{I_j}} \int\limits_{\eta^{\paren{1}}_j}^{\eta^{\paren{2}}_j} \paren{\eta_j - m_j} g_j\paren{\eta_1, \dots, \eta_j, \dots, \eta_k} \, d\eta_j \\
                                 &= 2 \ex{\eta_j}{\paren{\eta_j - m_j} g_j\paren{\eta}}.
\end{align*}
where $\paren{a}$ uses integration by parts that $m_j$ has been defined to be equal to $\frac{\eta^{\paren{1}}_j + \eta^{\paren{2}}_j}{2}$ in Section~\ref{subsec:intuition}.

Based on this, we remark that the inequality we wish to prove can be written in the form:
\begin{align}
    \ex{\eta_j, X}{2 \paren{\eta_j - m_j} f_j\paren{X} + \brk{f_j\paren{X} - \paren{\eta_j - m_j}}^2} \geq \frac{R_j^2}{12}. \label{eq:finger1}
\end{align}
Appealing to linearity of expectation, we focus on the second term and get:
\begin{align*}
    \ex{\eta_j, X}{\brk{f_j\paren{X} - \paren{\eta_j - m_j}}^2} &= \ex{\eta_j, X}{f_j^2\paren{X} - 2 \paren{\eta_j - m_j} f_j\paren{X} + \paren{\eta_j - m_j}^2} \\
                                                                &\geq - 2 \ex{\eta_j, X}{\paren{\eta_j - m_j} f_j\paren{X}} + \ex{\eta_j}{\paren{\eta_j - m_j}^2} \\
                                                                &= - 2 \ex{\eta_j, X}{\paren{\eta_j - m_j} f_j\paren{X}} + \var{\eta_j}{\eta_j} \\
                                                                &= - 2 \ex{\eta_j, X}{\paren{\eta_j - m_j} f_j\paren{X}} + \frac{R_j^2}{12}.
\end{align*}
Substituting this to (\ref{eq:finger1}) yields the desired result.
\end{proof}

We now generalize the previous lemma in the setting where all the components of $\eta$ are drawn independently and uniformly at random from intervals.

\begin{lemma}
\label{lem:fing_gen}
Let $p_{\eta}$ be a distribution over $S \subseteq \bR^d$ belonging to an exponential family $\cE\paren{T, h}$ with natural parameter vector $\eta \in \cH \subseteq \bR^k$.
We assume that $\eta$ is randomly generated by drawing independently $\eta_j \sim \cU\paren{I_j}, \forall j \in \brk{k}$.
Then, for any function $f \colon S^n \to \bigotimes\limits_{j \in \brk{k}} \brk{\pm \frac{R_j}{2}}$ that takes as input a dataset $X \coloneqq \paren{X_1, \dots, X_n} \sim p_{\eta}^{\bigotimes n}$, we have:
\[
    \ex{\eta, X}{Z + \llnorm{f\paren{X} - \paren{\eta - m}}^2} \geq \frac{\llnorm{R}^2}{12}.
\]
\end{lemma}

\begin{proof}
Without loss of generality, we assume that all the intervals $I_j$ are non-degenerate.
We have:
\begin{align*}
    \ex{\eta, X}{Z + \llnorm{f\paren{X} - \paren{\eta - m}}^2} &= \ex{\eta, X}{\sum\limits_{j \in \brk{k}} \brc{Z^j + \brk{f_j\paren{X} - \paren{\eta_j - m_j}}^2}} \\
                                                               &= \sum\limits_{j \in \brk{k}} \ex{\eta, X}{Z^j + \brk{f_j\paren{X} - \paren{\eta_j - m_j}}^2} \\
                                                               &= \sum\limits_{j \in \brk{k}} \ex{\eta_{- j}}{\ex{\eta_j, X}{Z^j + \brk{f_j\paren{X} - \paren{\eta_j - m_j}}^2}} \\
                                            \overset{\paren{a}}&{\geq} \sum\limits_{j \in \brk{k}} \ex{\eta_{- j}}{\frac{R_j^2}{12}} \\
                                                               &= \frac{\llnorm{R}^2}{12},
\end{align*}
where $\paren{a}$ is by a direct application of Lemma~\ref{lem:fing}.

If there exists a degenerate interval $I_j$, we have $Z^j = f_j\paren{X} - \paren{\eta_j - m_j} = R_j = 0$.
Thus, it suffices to apply the same argument as before for the non-degenerate intervals.
\end{proof}

\subsection{From the Lemma to Lower Bounds}
\label{subsec:LB}

We leverage Lemma~\ref{lem:fing_gen} to give a general recipe for proving lower bounds for private estimation.
The overall structure of our argument is similar to that of past fingerprinting proofs but, along the way, uses specific properties of exponential families to deal with the fact that we do not have a product distribution.

We assume the existence of an $\paren{\eps, \delta}$-DP mechanism $M \colon S^n \to \bigotimes\limits_{j \in \brk{k}} \brk{\pm \frac{R_j}{2}}$ with $\eps \in \brk{0, 1}$ and $\delta \geq 0$ such that, for any distribution $p_{\eta}$ with $\eta \in \bigotimes\limits_{j \in \brk{k}} I_j$, we have for $X \sim p_{\eta}^{\bigotimes n}$:
\[
    \ex{X, M}{\llnorm{M\paren{X} - \paren{\eta - m}}^2} \le \alpha^2 \le \frac{\llnorm{R}^2}{24}.
\]
Conditioning on the randomness of $M$, Lemma~\ref{lem:fing_gen} with $f \equiv M$ yields:
\begin{align}
    \ex{\eta, X, M}{Z} = \ex{M}{\ex{\eta, X}{Z}} \geq \frac{\llnorm{R}^2}{12} - \ex{\eta, X, M}{\llnorm{M\paren{X} - \paren{\eta - m}}^2} \geq \frac{\llnorm{R}^2}{24}. \label{eq:LB}
\end{align}
Thus, our intent is to upper-bound the LHS by a function of $n$ using the definition of privacy, which will allow us to identify a lower bound to get the desired accuracy.
To do so, we first need to upper-bound $\ex{\eta, X, M}{Z_i}$.
For that reason, in addition to the correlation terms (\ref{eq:corr_term1})-(\ref{eq:corr_term4}) defined in Section~\ref{subsec:intuition}, we introduce one more correlation term.
Let $X_i'$ be an independently drawn sample from $p_{\eta}$ and let $X_{\sim i}$ be the dataset that is the same as $X$ with $X_i$ replaced by $X_i'$.
We define:
\begin{align*}
                s_j &\coloneqq \brk{\frac{R_j^2}{4} - \paren{\eta_j - m_j}^2} \brk{M_j\paren{X_{\sim i}} - \paren{\eta_j - m_j}}, \forall j \in \brk{k}, \\
    \widetilde{Z}_i &\coloneqq \iprod{s, T\paren{X_i} - \mu_T} = \sum\limits_{j \in \brk{k}} s_j \paren{T_j\paren{X_i} - \mu_{T, j}}.
\end{align*}
This measures the effect that replacing an individual sample has on the output of the estimator.
Observe that, conditioning on $\eta$, $X_{\sim i}$ and $X_i$ are independent.
This implies that $\ex{X_{\sim i}, X_i, M}{\widetilde{Z}_i} = 0$ and $\var{X_{\sim i}, X_i, M}{\widetilde{Z}_i} = \ex{X_{\sim i}, X_i, M}{\widetilde{Z}_i^2}$.
Additionally, calculating the variance of $\widetilde{Z}_i$ requires us to reason about the covariance matrix of the sufficient statistics $\Sigma_T$, since the calculation involves terms of the form $\paren{T_{j_1}\paren{X_i} - \mu_{T, j_1}} \paren{T_{j_2}\paren{X_i} - \mu_{T, j_2}}, j_1, j_2 \in \brk{k}$.

The previous remarks are summarized in the following lemma.

\begin{lemma}
\label{lem:upp_bound1}
Conditioning on $\eta$, it holds that:
\begin{align*}
     \ex{X_{\sim i}, X_i, M}{\widetilde{Z}_i} &= 0, \\
    \var{X_{\sim i}, X_i, M}{\widetilde{Z}_i} &= \ex{X_{\sim i}, X_i, M}{\widetilde{Z}_i^2} = \ex{X_{\sim i}, M}{s^{\top} \Sigma_T s}.
\end{align*}
\end{lemma}

\begin{proof}
We start by computing the mean.
We have:
\begin{align*}
    \ex{X_{\sim i}, X_i, M}{\widetilde{Z}_i} = \ex{X_{\sim i}, X_i, M}{\sum\limits_{j \in \brk{k}} s_j \paren{T_j\paren{X_i} - \mu_{T, j}}} &= \sum\limits_{j \in \brk{k}} \ex{X_{\sim i}, X_i, M}{s_j \paren{T_j\paren{X_i} - \mu_{T, j}}} \\
                                                                                                                         \overset{\paren{a}}&{=} \sum\limits_{j \in \brk{k}} \ex{X_{\sim i}, M}{s_j} \cancelto{0}{\ex{X_i}{T_j\paren{X_i} - \mu_{T, j}}} \\
                                                                                                                                            &= 0,
\end{align*}
where $\paren{a}$ relies on the independence of $s_j = \brk{\frac{R_j^2}{4} - \paren{\eta_j - m_j}^2} \brk{M_j\paren{X_{\sim i}} - \paren{\eta_j - m_j}}$ and $T_j\paren{X_i} - \mu_{T, j}$ conditioned on $\eta$.

We now proceed with the variance:
\begin{align*}
    \var{X_{\sim i}, X_i, M}{\widetilde{Z}_i} &= \ex{X_{\sim i}, X_i, M}{\widetilde{Z}_i^2} \\
                                              &= \ex{X_{\sim i}, X_i, M}{\sum\limits_{j_1 \in \brk{k}} \sum\limits_{j_2 \in \brk{k}} s_{j_1} s_{j_2} \paren{T_{j_1}\paren{X_i} - \mu_{T, j_1}} \paren{T_{j_2}\paren{X_i} - \mu_{T, j_2}}} \\
                                              &= \sum\limits_{j_1 \in \brk{k}} \sum\limits_{j_2 \in \brk{k}} \ex{X_{\sim i}, M}{s_{j_1} s_{j_2}} \ex{X_i}{\paren{T_{j_1}\paren{X_i} - \mu_{T, j_1}} \paren{T_{j_2}\paren{X_i} - \mu_{T, j_2}}} \\
                           \overset{\paren{a}}&{=} \sum\limits_{j_1 \in \brk{k}} \sum\limits_{j_2 \in \brk{k}} \ex{X_{\sim i}, M}{s_{j_1} s_{j_2}} \Sigma_{T, j_1 j_2} \\
                                              &= \ex{X_{\sim i}, M}{\sum\limits_{j_1 \in \brk{k}} \sum\limits_{j_2 \in \brk{k}} s_{j_1} \Sigma_{T, j_1 j_2} s_{j_2}} \\
                                              &= \ex{X_{\sim i}, M}{s^{\top} \Sigma_T s}.
\end{align*}
where $\paren{a}$ is by the second part of Proposition~\ref{prop:exp_prop}.
\end{proof}

Now, we show how to upper-bound $\ex{X, M}{Z_i}$ using properties of differential privacy.
We appeal to the definition of DP to relate $\ex{X, M}{Z_i}$ with $\ex{X_{\sim i}, X_i, M}{\widetilde{Z}_i}$ and $\ex{X_{\sim i}, X_i, M}{\widetilde{Z}_i^2}$.
The proof involves splitting $Z_i$ in its positive and negative components and applying the definition of DP to them.
The main steps of the proof are standard (e.g., see~\cite{SteinkeU17b}), but we include everything for completeness.

\begin{lemma}
\label{lem:upp_bound2}
Conditioning on $\eta$, for any $\eps \in \brk{0, 1}, \delta \geq 0$, and $T > 0$, it holds that:
\[
    \ex{X, M}{Z_i} \le 2 \delta T + 2 \eps \sqrt{\ex{X_{\sim i}, M}{s^{\top} \Sigma_T s}} + 2 \int\limits_T^{\infty} \pr{X_i}{\llnorm{T\paren{X_i} - \mu_T} > \frac{4 t}{\infnorm{R}^3 \sqrt{k}}} \, dt.
\]
\end{lemma}

\begin{proof}
Consider the random variables $Z_{i, +} = \max\brc{Z_i, 0}$ and $Z_{i, -} = - \min\brc{Z_i, 0}$.
These random variables are non-negative, so for any $T > 0$, we have:
\begin{align}
    \ex{X, M}{Z_i} &= \ex{X, M}{Z_{i, +}} - \ex{\eta, X, M}{Z_{i, -}} \nonumber \\
                   &= \int\limits_0^{\infty} \pr{X, M}{Z_{i, +} > t} \, dt - \int\limits_0^{\infty} \pr{X, M}{Z_{i, -} > t} \, dt \nonumber \\
                   &= \int\limits_0^{\infty} \pr{X, M}{Z_i > t} \, dt - \int\limits_0^{\infty} \pr{X, M}{Z_i < - t} \, dt \nonumber \\
                   &= \int\limits_0^T \pr{X, M}{Z_i > t} \, dt - \int\limits_0^T \pr{X, M}{Z_i < - t} \, dt \nonumber \\
                   &\quad+ \int\limits_T^{\infty} \pr{X, M}{Z_i > t} \, dt - \int\limits_T^{\infty} \pr{X, M}{Z_i < - t} \, dt \label{eq:splitting1}
\end{align}
We focus our attention on the first two terms of (\ref{eq:splitting1}).
We consider the sets $S_{+, X_i, t}, S_{-, X_i, t} \subseteq \bR^k$:
\begin{align*}
    S_{+, X_i, t} &\coloneqq \brc{Y \in \bR^k \colon \sum\limits_{j \in \brk{k}} \brk{\frac{R_j^2}{4} - \paren{\eta_j - m_j}^2} \brk{Y_j - \paren{\eta_j - m_j}} \paren{T_j\paren{X_i} - \mu_{T, j}} \geq t}, \\
    S_{-, X_i, t} &\coloneqq \brc{Y \in \bR^k \colon \sum\limits_{j \in \brk{k}} \brk{\frac{R_j^2}{4} - \paren{\eta_j - m_j}^2} \brk{Y_j - \paren{\eta_j - m_j}} \paren{T_j\paren{X_i} - \mu_{T, j}} \le - t}.
\end{align*}
By a direct application of Definition~\ref{def:dp} on $S_{+, X_i, t}$ and $S_{-, X_i, t}$\footnote{These are random subsets of $\bR^k$, since they depend on the random vector $X_i$ ($\eta$ is assumed to be fixed).
This does not affect the way the definition of privacy is applied.} for the neighboring datasets $X$ and $X_{\sim i}$ we have:
\begin{align*}
    &\quad\ \int\limits_0^T \pr{X, M}{Z_i > t} \, dt - \int\limits_0^T \pr{X, M}{Z_i < - t} \, dt \\
    &= \int\limits_0^T \pr{X, M}{M\paren{X} \in S_{+, X_i, t}} \, dt - \int\limits_0^T \pr{X, M}{M\paren{X} \in S_{-, X_i, t}} \, dt \\
    &\le \int\limits_0^T \paren{e^{\eps} \pr{X_{\sim i}, X_i, M}{M\paren{X_{\sim i}} \in S_{+, X_i, t}} + \delta} \, dt \\
    &\quad- \int\limits_0^T e^{- \eps} \paren{\pr{X_{\sim i}, X_i, M}{M\paren{X_{\sim i}} \in S_{-, X_i, t}} - \delta} \, dt \\
    &\le \int\limits_0^T \paren{e^{\eps} \pr{X_{\sim i}, X_i, M}{\widetilde{Z}_i > t} + \delta} \, dt - \int\limits_0^T e^{- \eps} \paren{\pr{X_{\sim i}, X_i, M}{\widetilde{Z}_i < - t} - \delta} \, dt \\
    &= \paren{1 + e^{- \eps}} \delta T + e^{\eps} \int\limits_0^T \pr{X_{\sim i}, X_i, M}{\widetilde{Z}_i > t} \, dt - e^{- \eps} \int\limits_0^T\pr{X_{\sim i}, X_i, M}{\widetilde{Z}_i < - t} \, dt.
\end{align*}
Under the assumption that $\eps \in \brk{0, 1}$, it holds that $e^{\eps} \le 1 + 2 \eps$ and $1 \geq e^{- \eps} \geq 1 - \eps \geq 1 - 2 \eps$.
Combined with the above, this yields:
\begin{align*}
    &\quad\ \int\limits_0^T \pr{X, M}{Z_i > t} \, dt - \int\limits_0^T \pr{X, M}{Z_i < - t} \, dt \\
    &\le 2 \delta T + \paren{2 \eps + 1} \int\limits_0^T \pr{X_{\sim i}, X_i, M}{\widetilde{Z}_i > t} \, dt + \paren{2 \eps - 1} \int\limits_0^T \pr{X_{\sim i}, X_i, M}{\widetilde{Z}_i < - t} \, dt \\
    &= 2 \delta T + 2 \eps \int\limits_0^T \paren{\pr{X_{\sim i}, X_i, M}{\widetilde{Z}_i > t} + \pr{X_{\sim i}, X_i, M}{\widetilde{Z}_i < - t}} \, dt \\
    &\quad+ \int\limits_0^T \paren{\pr{X_{\sim i}, X_i, M}{\widetilde{Z}_i > t} - \pr{X_{\sim i}, X_i, M}{\widetilde{Z}_i < - t}} \, dt \\
    &= 2 \delta T \nonumber + 2 \eps \int\limits_0^T \pr{X_{\sim i}, X_i, M}{\abs{\widetilde{Z}_i} > t} \, dt + \int\limits_0^T \paren{\pr{X_{\sim i}, X_i, M}{\widetilde{Z}_i > t} - \pr{X_{\sim i}, X_i, M}{\widetilde{Z}_i < - t}} \, dt \\
    &\le 2 \delta T \nonumber + 2 \eps \int\limits_0^{\infty} \pr{X_{\sim i}, X_i, M}{\abs{\widetilde{Z}_i} > t} \, dt + \int\limits_0^T \paren{\pr{X_{\sim i}, X_i, M}{\widetilde{Z}_i > t} - \pr{X_{\sim i}, X_i, M}{\widetilde{Z}_i < - t}} \, dt \\
    &= 2 \delta T + 2 \eps \ex{X_{\sim i}, X_i, M}{\abs{\widetilde{Z}_i}} + \int\limits_0^T \paren{\pr{X_{\sim i}, X_i, M}{\widetilde{Z}_i > t} - \pr{X_{\sim i}, X_i, M}{\widetilde{Z}_i < - t}} \, dt \\
    &= 2 \delta T + 2 \eps \ex{X_{\sim i}, X_i, M}{\abs{\widetilde{Z}_i}} + \int\limits_0^{\infty} \paren{\pr{X_{\sim i}, X_i, M}{\widetilde{Z}_i > t} - \pr{X_{\sim i}, X_i, M}{\widetilde{Z}_i < - t}} \, dt \\
    &\quad- \int\limits_T^{\infty} \paren{\pr{X_{\sim i}, X_i, M}{\widetilde{Z}_i > t} - \pr{X_{\sim i}, X_i, M}{\widetilde{Z}_i < - t}} \, dt.
\end{align*}
Defining $\widetilde{Z}_{i, +} = \max\brc{\widetilde{Z}_i, 0}$ and $\widetilde{Z}_{i, -} = - \min\brc{\widetilde{Z}_i, 0}$, we get:
\begin{align}
    &\quad\ \int\limits_0^T \pr{X, M}{Z_i > t} \, dt - \int\limits_0^T \pr{X, M}{Z_i < - t} \, dt \nonumber \\
    &\le 2 \delta T + 2 \eps \ex{X_{\sim i}, X_i, M}{\abs{\widetilde{Z}_i}} + \paren{\ex{X_{\sim i}, X_i, M}{\widetilde{Z}_{i, +}} - \ex{X_{\sim i}, X_i, M}{\widetilde{Z}_{i, -}}} \nonumber \\
    &\quad- \int\limits_T^{\infty} \paren{\pr{X_{\sim i}, X_i, M}{\widetilde{Z}_i > t} - \pr{X_{\sim i}, X_i, M}{\widetilde{Z}_i < - t}} \, dt \nonumber \\
    &\le 2 \delta T + 2 \eps \ex{X_{\sim i}, X_i, M}{\abs{\widetilde{Z}_i}} + \cancelto{0}{\ex{X_{\sim i}, X_i, M}{\widetilde{Z}_i}} \nonumber \\
    &\quad- \int\limits_T^{\infty} \paren{\pr{X_{\sim i}, X_i, M}{\widetilde{Z}_i > t} - \pr{X_{\sim i}, X_i, M}{\widetilde{Z}_i < - t}} \, dt \nonumber \\
    &= 2 \delta T + 2 \eps \ex{X_{\sim i}, X_i, M}{\abs{\widetilde{Z}_i}} - \int\limits_T^{\infty} \paren{\pr{X_{\sim i}, X_i, M}{\widetilde{Z}_i > t} - \pr{X_{\sim i}, X_i, M}{\widetilde{Z}_i < - t}} \, dt, \label{eq:splitting2}
\end{align}
where we appealed to Lemma~\ref{lem:upp_bound1} along the way.

The second version of Fact~\ref{fact:cs} for $X = \abs{\widetilde{Z}_i}, Y = 1$, and Lemma~\ref{lem:upp_bound1} yield:
\begin{align}
    \ex{X_{\sim i}, X_i, M}{\abs{\widetilde{Z}_i}} \le \sqrt{\ex{X_{\sim i}, X_i, M}{\widetilde{Z}_i^2}} = \sqrt{\ex{X_{\sim i}, M}{s^{\top} \Sigma_T s}} \label{eq:splitting3}.
\end{align}
Substituting (\ref{eq:splitting2}) and (\ref{eq:splitting3}) into (\ref{eq:splitting1}), we get:
\begin{align}
    \ex{X, M}{Z_i} &\le 2 \delta T + 2 \eps \sqrt{\ex{X_{\sim i}, M}{s^{\top} \Sigma_T s}} + \int\limits_T^{\infty} \paren{\pr{X, M}{Z_i > t} - \pr{X, M}{Z_i < - t}} \, dt \nonumber \\
                   &\quad+ \int\limits_T^{\infty} \paren{\pr{X_{\sim i}, X_i, M}{\widetilde{Z}_i < - t} - \pr{X_{\sim i}, X_i, M}{\widetilde{Z}_i > t}} \, dt \nonumber \\
                   &\le 2 \delta T + 2 \eps \sqrt{\ex{X_{\sim i}, M}{s^{\top} \Sigma_T s}} + \int\limits_T^{\infty} \paren{\pr{X, M}{\abs{Z_i} > t} + \pr{X_{\sim i}, X_i, M}{\abs{\widetilde{Z}_i} > t}} \, dt. \label{eq:splitting4}
\end{align}
We now show how to bound each of the tail probabilities involving $Z_i$ and $\widetilde{Z}_i$.
For $Z_i$, we have:
\begin{align*}
    \abs{Z_i} &= \abs{\sum\limits_{j \in \brk{k}} \brk{\frac{R_j^2}{4} - \paren{\eta_j - m_j}^2} \brk{M_j\paren{X} - \paren{\eta_j - m_j}} \paren{T_j\paren{X_i} - \mu_{T, j}}} \\
              &\le \sum\limits_{j \in \brk{k}} \brk{\frac{R_j^2}{4} - \paren{\eta_j - m_j}^2} \abs{M_j\paren{X} - \paren{\eta_j - m_j}} \abs{T_j\paren{X_i} - \mu_{T, j}} \\
    \overset{\paren{a}}&{\le} \frac{\infnorm{R}^3}{4} \sum\limits_{j \in \brk{k}} \abs{T_j\paren{X_i} - \mu_{T, j}} \\
    \overset{\paren{b}}&{\le} \frac{\infnorm{R}^3}{4} \sqrt{k} \llnorm{T\paren{X_i} - \mu_T},
\end{align*}
where $\paren{a}$ relies on the assumption $\abs{M_j\paren{X_i} - \paren{\eta_j - m_j}} \le R_j, \forall j \in \brk{k}$, and $\paren{b}$ uses the first version of Fact~\ref{fact:cs} for $x = \Vec{1}_k, y_j = \abs{T_j\paren{X_i} - \mu_{T, j}}, \forall j \in \brk{k}$.

The above yields:
\begin{align}
    \int\limits_T^{\infty} \pr{X, M}{\abs{Z_i} > t} \, dt \le \int\limits_T^{\infty} \pr{X_i}{\llnorm{T\paren{X_i} - \mu_T} > \frac{4 t}{\infnorm{R}^3 \sqrt{k}}} \, dt. \label{eq:tail_prob_equiv1}
\end{align}
Upper-bounding the term involving $\widetilde{Z}_i$ in similar fashion and substituting (\ref{eq:tail_prob_equiv1}) to (\ref{eq:splitting4}), we get:
\[
    \ex{X, M}{Z_i} \le 2 \delta T + 2 \eps \sqrt{\ex{X_{\sim i}, M}{s^{\top} \Sigma_T s}} + 2 \int\limits_T^{\infty} \pr{X_i}{\llnorm{T\paren{X_i} - \mu_T} > \frac{4 t}{\infnorm{R}^3 \sqrt{k}}} \, dt.
\]
\end{proof}

We conclude by formally stating and proving the main result of this section and describing the general recipe about how to apply it.

\begin{theorem}
\label{thm:lower_bound}
Let $p_{\eta}$ be a distribution over $S \subseteq \bR^d$ belonging to an exponential family $\cE\paren{T, h}$ with natural parameter vector $\eta \in \cH \subseteq \bR^k$.
Also, let $\eta^{\paren{1}}, \eta^{\paren{2}} \in \cH$ and let $I_j \coloneqq \brk{\eta^{\paren{1}}_j, \eta^{\paren{2}}_j}, \forall j \in \brk{k}$ be a collection of intervals and $R \coloneqq \eta^{\paren{2}} - \eta^{\paren{1}}, m \coloneqq \frac{\eta^{\paren{1}} + \eta^{\paren{2}}}{2}$ be the corresponding width and midpoint vectors, respectively.
Assume that $\bigotimes\limits_{j \in \brk{k}} I_j \subseteq \cH$ and $\eta \sim \cU\paren{\bigotimes\limits_{j \in \brk{k}} I_j}$.
Moreover, assume that we have a dataset $X \sim p_{\eta}^{\bigotimes n}$ and an independently drawn point $X_i' \sim p_{\eta}$ and $X_{\sim i}$ denotes the dataset where $X_i$ has been replaced with $X_i'$.
Finally, let $M \colon S^n \to \bigotimes\limits_{j \in \brk{k}} \brk{\pm \frac{R_j}{2}}$ be an $\paren{\eps, \delta}$-DP mechanism with $\eps \in \brk{0, 1}, \delta \geq 0$ with $\ex{X, M}{\llnorm{M\paren{X} - \paren{\eta - m}}^2} \le \alpha^2 \le \frac{\llnorm{R}^2}{24}$.
Then, for any $T > 0$, it holds that:
\[
    n \brc{2 \delta T + 2 \eps \ex{\eta}{\sqrt{\ex{X_{\sim i}, M}{s^{\top} \Sigma_T s}}} + 2 \ex{\eta}{\int\limits_T^{\infty} \pr{X_i}{\llnorm{T\paren{X_i} - \mu_T} > \frac{4 t}{\infnorm{R}^3 \sqrt{k}}} \, dt}} \geq \frac{\llnorm{R}^2}{24}, \label{eq:main_ineq}
\]
where $s \in \bR^k$ with:
\[
    s_j \coloneqq \brk{\frac{R_j^2}{4} - \paren{\eta_j - m_j}^2} \brk{M_j\paren{X_{\sim i}} - \paren{\eta_j - m_j}}, \forall j \in \brk{k}.
\]
\end{theorem}

\begin{proof}
By Lemma~\ref{lem:fing_gen} and (\ref{eq:LB}), we have $\ex{\eta, X, M}{Z} \geq \frac{\llnorm{R}^2}{24}$.
Additionally, by Lemma~\ref{lem:upp_bound2}, we have:
\begin{align*}
    \ex{\eta, X, M}{Z} &= \sum\limits_{i \in \brk{n}} \ex{\eta, X, M}{Z_i} \\
                       &= \sum\limits_{i \in \brk{n}} \ex{\eta}{\ex{X, M}{Z_i}} \\
                       &\le \sum\limits_{i \in \brk{n}} \ex{\eta}{2 \delta T + 2 \eps \sqrt{\ex{X_{\sim i}, M}{s^{\top} \Sigma_T s}} + 2 \int\limits_T^{\infty} \pr{X_i}{\llnorm{T\paren{X_i} - \mu_T} > \frac{4 t}{\infnorm{R}^3 \sqrt{k}}} \, dt} \\
                       &\le n \brc{2 \delta T + 2 \eps \ex{\eta}{\sqrt{\ex{X_{\sim i}, M}{s^{\top} \Sigma_T s}}} + 2 \ex{\eta}{\int\limits_T^{\infty} \pr{X_i}{\llnorm{T\paren{X_i} - \mu_T} > \frac{4 t}{\infnorm{R}^3 \sqrt{k}}} \, dt}}.
\end{align*}
Combining the two, we get the desired result.
\end{proof}

In order to apply the above theorem, the main idea is to identify values for $T$ and $\delta$ such that:
\[
    \delta T \geq 2 \ex{\eta}{\int\limits_T^{\infty} \pr{X_i}{\llnorm{T\paren{X_i} - \mu_T} > \frac{4 t}{\infnorm{R}^3 \sqrt{k}}} \, dt} \text{ and } 3 \delta T n \le \frac{\llnorm{R}^2}{48}.
\]
The reasoning behind this is that the first and third terms of (\ref{eq:main_ineq}) both depend on the threshold $T$ (below which the definition of privacy is applied), with the latter term becoming smaller as $T$ increases.
Thus, balancing these two terms involves identifying a value for $T$ that is as small as possible and results in $\delta T$ dominating the other term.
At the same time, we want the sum of those two terms to be smaller than the half of the RHS, so that we get a lower bound that is as tight as possible, leading to the constraint on $\delta$.
The previous result in the inequality $n \ex{\eta}{\sqrt{\ex{X_{\sim i}, M}{s^{\top} \Sigma_T s}}} \geq \frac{\llnorm{R}^2}{96 \eps}$.
Consequently, obtaining the desired sample complexity lower bounds boils down to upper-bounding the term $\ex{X_{\sim i}, M}{s^{\top} \Sigma_T s}$.
In the applications we will see, the upper bounds will be worst-case in terms of $\eta$, so the outer expectation $\ex{\eta}{\cdot}$ will largely be ignored.

\begin{remark}
\label{rem:hyperrect}
In order to apply Theorem~\ref{thm:lower_bound}, it is necessary for the parameter range $\cH$ to have a subset that is a (potentially degenerate) axis-aligned hyper-rectangle.
This is a very mild assumption, since $\cH$ is known to be convex (see the last part of Proposition~\ref{prop:exp_prop}).
Moreover, if we have identified a hyper-rectangle that is subset of $\cH$ and is convenient for our purposes, but is not axis-aligned, it suffices to remark that, for any rotation matrix $U \in \bR^{k \times k}, \eta^{\top} T\paren{x} = \paren{U \eta}^{\top} \paren{U T\paren{x}}$.
This mapping is bijective and preserves the $\ell_2$-distance between vectors.
Thus, we can re-parameterize our exponential family in a way that will make the aforementioned hyper-rectangle axis-aligned and lower bounds proven for estimating the natural parameter vector of the new family will hold for the original one as well.
\end{remark}

\begin{remark}
\label{rem:exp_chernoff}
The last term in the upper bound of Lemma~\ref{lem:upp_bound2} involves the tail of the random variable $\llnorm{T\paren{X_i} - \mu_T}$.
The obvious way to bound that term is:
\begin{align*}
    \int\limits_T^{\infty} \pr{X_i}{\llnorm{T\paren{X_i} - \mu_T}^2 > \frac{16 t^2}{\infnorm{R}^6 k}} \, dt \le \bigintss\limits_T^{\infty} \frac{\ex{X_i}{\llnorm{T\paren{X_i} - \mu_T}^2}}{\frac{16 t^2}{\infnorm{R}^6 k}} \, dt &= \frac{\infnorm{R}^6 k \tr\paren{\Sigma_T}}{16} \int\limits_T^{\infty} \frac{dt}{t^2} \\
                                                                             &= \frac{\infnorm{R}^6 k \tr\paren{\Sigma_T}}{16 T}.
\end{align*}
Following the steps we sketched previously, this approach leads to a range for $\delta$ of the form $\delta \le \widetilde{\cO}\paren{\frac{1}{n^2}}$.
However, when reasoning about the cost of approximate DP for various statistical tasks, in order to claim that a full characterization has been achieved, it is generally the case that lower bounds must be shown for $\delta$ as large as $\widetilde{\cO}\paren{\frac{1}{n}}$.
To achieve this, we will apply the following strategy.
First, observe that we can approximate $\llnorm{T\paren{X} - \mu_T}$ using an $\eps$-net (see Definition~\ref{def:net}).
In particular, if $\cN_{\frac{1}{4}}$ is a $\frac{1}{4}$-net of $\bS^{k - 1}$ of minimum cardinality, using Facts~\ref{fact:cov_num_unit_sphere},~\ref{fact:l2_norm_bound}, and a union bound, we get that:
\begin{align*}
    \int\limits_T^{\infty} \pr{X_i}{\llnorm{T\paren{X_i} - \mu_T} > \frac{4 t}{\infnorm{R}^3 \sqrt{k}}} \, dt &\le \int\limits_T^{\infty} \pr{X_i}{\frac{1}{1 - \frac{1}{4}} \cdot \sup\limits_{v \in \cN_{\frac{1}{4}}} \abs{\iprod{v, T\paren{X_i} - \mu_T}} > \frac{4 t}{\infnorm{R}^3 \sqrt{k}}} \, dt \\
                                                                                                              &= \int\limits_T^{\infty} \pr{X_i}{\sup\limits_{v \in \cN_{\frac{1}{4}}} \abs{\iprod{v, T\paren{X_i} - \mu_T}} > \frac{3 t}{\infnorm{R}^3 \sqrt{k}}} \, dt \\
                                                                                                              &\le \abs{\cN_{\frac{1}{4}}} \int\limits_T^{\infty} \pr{X_i}{\abs{\iprod{v, T\paren{X_i} - \mu_T}} > \frac{3 t}{\infnorm{R}^3 \sqrt{k}}} \, dt \\
                                                                                                              &\le 9^k \sup_{v \in \cN_{\frac{1}{4}}} \int\limits_T^{\infty} \pr{X_i}{\abs{\iprod{v, T\paren{X_i} - \mu_T}} > \frac{3 t}{\infnorm{R}^3 \sqrt{k}}} \, dt.
\end{align*}
Now, to reason about the tail probabilities of $\abs{\iprod{v, T\paren{X_i} - \mu_T}}$, it suffices to leverage the fact that the MGFs of univariate projections of the sufficient statistics of exponential families are well-defined (see the third part of Proposition~\ref{prop:exp_prop}).
This paves the way for proving Chernoff-type bounds in the context that we are considering.
In practice, we will not have to prove our own bounds, instead appealing to various known bounds from the literature.
\end{remark}

\begin{remark}
\label{rem:fing_pure_dp}
In the above discussion, we have largely assumed that we want to prove lower bounds under approximate DP, corresponding to $\delta > 0$.
However, applying the method becomes significantly simpler for pure DP.
Indeed, the reason we had to split $\bR_+$ in the intervals $\brk{0, T}$ and $\paren{T, \infty}$ in the proof of Lemma~\ref{lem:upp_bound2} and work separately with each interval is because, if we applied the definition of DP to each of $\pr{X, M}{Z_i > t}, \pr{X, M}{Z_i < - t}$ while integrating over all of $\bR_+$, we would have to deal with a term of the form $\int\limits_0^{\infty} \delta \, dt = \infty$.
This would render our upper bound on the correlation terms vacuous and it would be impossible to obtain any meaningful lower bounds.
On the other hand, for $\delta = 0$, it suffices to integrate over all of $\bR_+$ (corresponding to $T = \infty$).
This would eventually lead to the bound $n \ex{\eta}{\sqrt{\ex{X_{\sim i}, M}{s^{\top} \Sigma_T s}}} \geq \frac{\llnorm{R}^2}{48 \eps}$.
\end{remark}

\section{Lower Bounds for Private Gaussian Covariance Estimation}
\label{sec:lb_cov_est}

Here, we provide lower bounds for covariance estimation of high-dimensional Gaussians with respect to the Mahalanobis and spectral norms under the constraint of $\paren{\eps, \delta}$-DP.

\subsection{Estimation with Respect to the Mahalanobis Norm}
\label{subsec:mah_norm_lb}

In this section, we characterize the sample complexity of private Gaussian covariance estimation with respect to the Mahalanobis norm under approximate differential privacy.
Previous lower bounds either assumed a stricter notion of privacy (e.g., $\eps$-DP - see~\cite{KamathLSU19}) or exhibited significant gaps with known upper bounds (see Section $1.1.1$ of~\cite{AdenAliAK21} for further discussion).
We start by stating our main result.

\begin{theorem}[Gaussian Covariance Estimation in Mahalanobis Norm]\label{thm:lb-maha}
There exists a distribution $\cD$ over covariance matrices $\Sigma \in \bR^{d \times d}$ with $\id \preceq \Sigma \preceq 2 \id$ such that, given $\Sigma \sim \cD$ and $X \sim \cN\paren{0, \Sigma}^{\bigotimes n}$, for any $\alpha = \cO\paren{1}$ and any $\paren{\eps, \delta}$-DP mechanism $M \colon \bR^{n \times d} \to \bR^{d \times d}$ with $\eps, \delta \in \brk{0, 1}$, and $\delta \le \cO\paren{\min\brc{\frac{1}{n}, \frac{d^2}{n \log\paren{\frac{n}{d^2}}}}}$ that satisfies $\ex{X, M}{\mnorm{\Sigma}{M\paren{X} - \Sigma}^2} \le \alpha^2$, it must hold that $n \geq \Omega\paren{\frac{d^2}{\alpha \eps}}$.
\end{theorem}

The outline of the proof involves:
\begin{enumerate}
    \item writing the Gaussian distribution $\cN\paren{0, \Sigma}$ as an exponential family (Lemma~\ref{lem:exp_fam_norm_cov_alt}),
    \item defining a distribution from which $\eta$ is drawn (Algorithm~\ref{alg:sampling} and Lemma~\ref{lem:sampling_analysis}),
    \item showing that estimating the deviation of the natural parameter vector from the midpoint $m$ reduces to covariance estimation (so lower bounds for the former also apply for the latter - see Algorithm~\ref{alg:reduction}, as well as Lemma~\ref{lem:reduction_cov} and Corollary~\ref{cor:red_cov_lb}),
    \item upper-bounding the term $\ex{}{s^{\top} \Sigma_T s}$ (Lemma~\ref{lem:suff_stats_cov_ub}),
    \item applying Theorem~\ref{thm:lower_bound} as described at the end of the previous section (Lemma~\ref{lem:tail_prob_cov_ub1} and the main proof at the end of the section).
\end{enumerate}

Regarding the first step, Fact~\ref{fact:exp_fam_norm_cov} establishes that the density can be written in the form:
\begin{align*}
    p_{\eta}\paren{x} &= h\paren{x} \exp\paren{\eta^{\top} T\paren{x} - Z\paren{\eta}}, \forall x \in \bR^d, \\
           h\paren{x} &= 1, \\
           T\paren{x} &= - \frac{1}{2} \paren{x x^{\top}}^{\flat}, \\
                 \eta &= \paren{\Sigma^{- 1}}^{\flat}, \\
        Z\paren{\eta} &= \frac{d}{2} \ln\paren{2 \pi} - \frac{1}{2} \ln\paren{\det\paren{\eta^{\#}}}.
\end{align*}
Observe that the natural parameter vector $\eta$ is the flattening of the \emph{precision matrix} $\Sigma^{-1}$.
However, we cannot use this representation, because $\Sigma^{- 1}$ is symmetric.
Indeed, this implies that the components of $\eta$ cannot be generated independently, as is required by Theorem~\ref{thm:lower_bound}.
For that reason, we consider an alternative parameterization.

\begin{lemma}
\label{lem:exp_fam_norm_cov_alt}
Let $\cN\paren{0, \Sigma}$ be a $d$-dimensional Gaussian with unknown covariance $\Sigma$.
Then, it belongs to an exponential family with $h_0 \equiv h, \eta_0 \coloneqq 2 U^{\flat}, T_0 \equiv T$, and $Z_0\paren{\eta_0} \coloneqq \frac{d}{2} \ln\paren{2 \pi} - \frac{1}{2} \ln\paren{\det\paren{\frac{\eta_0^{\#} + \paren{\eta_0^{\#}}^{\top}}{2}}}$ where $U \in \bR^{d \times d}$ is an upper-triangular matrix with:
\[
    U_{i j} \coloneqq \left\{
    \begin{array}{ll}
         \frac{1}{2} \paren{\Sigma^{- 1}}_{i i}, & i = j \\
         \paren{\Sigma^{- 1}}_{i j},             & j > i \\
         0,                                      & i > j
    \end{array}
    \right..
\]
In particular, it holds that $Z\paren{\eta} = Z_0\paren{\eta_0}$ and $\eta^{\top} T\paren{x} = \eta_0^{\top} T_0\paren{x}$.
\end{lemma}

\begin{proof}
Observe that $\eta^{\#} = \Sigma^{- 1} = U + U^{\top} = \frac{\eta_0^{\#} + \paren{\eta_0^{\#}}^{\top}}{2} \implies Z\paren{\eta} = Z_0\paren{\eta_0}$.
Additionally, we have:
\begin{align*}
    \eta^{\top} T\paren{x} = \iprod{\paren{\Sigma^{- 1}}^{\flat}, - \frac{1}{2} \paren{x x^{\top}}^{\flat}} = - \frac{1}{2} \iprod{\Sigma^{- 1}, x x^{\top}} &= - \frac{1}{2} x^{\top} \Sigma^{- 1} x \\
                                                                                                                                                             &= - \frac{1}{2} x^{\top} \paren{U + U^{\top}} x \\
                                                                                                                                                             &= - \frac{1}{2} x^{\top} U x - \frac{1}{2} x^{\top} U^{\top} x \\
                                                                                                                                          \overset{\paren{a}}&{=} - x^{\top} U x \\
                                                                                                                                                             &= - \iprod{U, x x^{\top}} \\
                                                                                                                                                             &= \iprod{2 U^{\flat}, - \frac{1}{2} \paren{x x^{\top}}^{\flat}} \\
                                                                                                                                                             &= \eta_0^{\top} T_0\paren{x},
\end{align*}
where $\paren{a}$ used the fact that $x^{\top} U^{\top} x \in \bR \implies x^{\top} U^{\top} x = \paren{x^{\top} U^{\top} x}^{\top} = x^{\top} U x$.

Consequently, we have $p_{\eta}\paren{x} = h\paren{x} \exp\paren{\eta^{\top} T\paren{x} - Z\paren{\eta}} = h_0\paren{x} \exp\paren{\eta_0^{\top} T_0\paren{x} - Z_0\paren{\eta_0}}$, verifying the equivalence of the two parameterizations.
\end{proof}

We now consider the following process to generate $\eta_0$.
Roughly speaking, modulo the parameterization issues mentioned above, the process consists of sampling the elements of the precision matrix uniformly at random from intervals of width $\frac{1}{2 d}$.
For the diagonal elements, these intervals are centered at $\frac{3}{4}$, whereas for the non-diagonal elements they are centered at the origin.
Similar constructions have been employed previously, e.g., in~\cite{KamathLSU19}.

\begin{algorithm}[htb]
    \caption{Covariance Matrix Sampling}\label{alg:sampling}
    \hspace*{\algorithmicindent} \textbf{Input:} $d \in \N$. \\
    \hspace*{\algorithmicindent} \textbf{Output:} A pair $\paren{\Sigma, \eta_0}$.
    \begin{algorithmic}[1]
    \Procedure{\CovSampling}{$d$}
        \For{$i \in \brk{d}$}
            \For{$j \in \brk{d} \setminus \brk{i - 1}$}
                \If {$i = j$}
                    \State Draw an independent sample $\paren{\Sigma^{- 1}}_{i i} \sim \cU\brk{\frac{3}{4} \pm \frac{1}{4 d}}$.
                    \State Let $U_{i i} = \frac{1}{2} \paren{\Sigma^{- 1}}_{i i}$.
                \Else
                    \State Draw an independent sample $\paren{\Sigma^{- 1}}_{i j} \sim \cU\brk{\pm \frac{1}{4 d}}$.
                    \State Let $\paren{\Sigma^{- 1}}_{j i} = \paren{\Sigma^{- 1}}_{i j}$.
                    \State Let $U_{i j} = \paren{\Sigma^{- 1}}_{i j}$ and $U_{j i} = 0$.
                \EndIf
            \EndFor
        \EndFor
        \State Let $\eta_0 = 2 U^{\flat}$.
        \State \Return $\paren{\Sigma, \eta_0}$.
    \EndProcedure
    \end{algorithmic}
\end{algorithm}

We proceed to identify some properties of the matrices generated by the above process that will help us to prove Theorem~\ref{thm:lb-maha}.

\begin{lemma}
\label{lem:sampling_analysis}
Let $\paren{\Sigma, \eta_0} \in \bR^{d \times d} \times \bR^{d^2}$ be the output of Algorithm~\ref{alg:sampling}.
We have that:
\begin{enumerate}
    \item $\id \preceq \Sigma \preceq 2 \id$.
    \item The components of $\eta_0$ are supported on intervals of the form:
    \[
        I_{\ell} = \left\{
        \begin{array}{ll}
            \brk{\frac{3}{4} \pm \frac{1}{4 d}},    & \paren{\ell - 1} \bmod d = \paren{\ell - 1} \bdiv d \\
            \brk{\pm \frac{1}{2 d}},                & \paren{\ell - 1} \bmod d > \paren{\ell - 1} \bdiv d \\
            \brc{0},                                & \paren{\ell - 1} \bmod d < \paren{\ell - 1} \bdiv d
        \end{array}
        \right..
    \]
    \item The components of $\eta_0$ are independent and it holds that $\infnorm{R} = \frac{1}{d}$ and $\llnorm{R}^2 = \frac{1}{2} \paren{1 - \frac{1}{2 d}} \geq \frac{1}{4}$.
    \item The vector $m \in \bR^{d^2}$ is equal to:
    \[
        m_{\ell} = \left\{
        \begin{array}{ll}
            \frac{3}{4}, & \ell = \paren{k - 1} d + k,\text{ for some } k \in \brk{d} \\
            0,           & \text{otherwise} 
        \end{array}
        \right..
    \]
\end{enumerate}
\end{lemma}

\begin{proof}
\begin{enumerate}
    \item By a direct application of Theorem~\ref{thm:gershg}, we get that the eigenvalues of $\Sigma^{- 1}$ all lie in the interval $\brk{\frac{1}{2}, 1}$.
    Thus, we have $\frac{1}{2} \id \preceq \Sigma^{- 1} \preceq \id \iff \id \preceq \Sigma \preceq 2 \id$.
    \item Observe that, the first of the three branches in the definition corresponds to the diagonal elements of $\eta_0^{\#}$, the second corresponds to those above the diagonal, and the third corresponds to those below the diagonal.
    The result follows directly from this remark.
    \item The non-constant components of $\eta_0$ are all independent draws from a uniform distribution, which immediately yields the independence property.
    In addition to that, we have $d$ non-constant components taking values in an interval of length $\frac{1}{2 d}$, and $\frac{d \paren{d - 1}}{2}$ non-constant components take values in an interval of length $\frac{1}{d}$.
    This implies that $\infnorm{R} = \frac{1}{d}$ and $\llnorm{R}^2 = d \cdot \frac{1}{4 d^2} + \frac{d \paren{d - 1}}{2} \cdot \frac{1}{d^2} = \frac{1}{2} \paren{1 - \frac{1}{2 d}}$.
    \item Observe that the non-diagonal elements of $\eta_0^{\#} = 2 U$ are either constant and equal to $0$ or take values in a $0$-centered interval.
    This implies that the only non-zero elements of $m$ are those corresponding to the diagonal elements of $\paren{\Sigma^{- 1}}_{i i}$.
    This yields the desired result.
\end{enumerate}
\end{proof}

\begin{remark}
\label{rem:accuracy}
Note that Theorem~\ref{thm:lb-maha} holds only for $\alpha = \cO\paren{1}$.
In contrast, prior fingerprinting lower bounds for mean estimation of product distributions and Gaussians gave non-trivial results in the low-accuracy regime $\alpha = \cO\paren{\sqrt{d}}$ (see, e.g., the results in Section~\ref{sec:existing}).
In principle, it may be able to achieve similar lower bounds for covariance estimation in Frobenius norm, as the Frobenius diameter of the set of matrices that satisfy $\id \preceq \Sigma \preceq 2 \id$ is equal to $\sqrt{d}$.
However, by the third part of the above lemma, our construction only has an upper bound on this diameter of $\cO\paren{1}$.
This makes estimation trivial for any $\alpha$ which is larger, as one could output an arbitrary parameter vector within the set.
Thus, to prove a lower bound for all $\alpha = \cO\paren{\sqrt{d}}$, one must consider a different construction, and it is not merely a deficiency of our analysis.
\end{remark}

We show that, if one has an estimator for the covariance matrix $\Sigma$, then this implies an estimator for the natural parameter vector $\eta_0$.
Therefore, a lower bound for the latter problem implies a lower bound for the former, allowing us to focus on lower bounds for estimating the natural parameter vector.
More precisely, we present an $\paren{\eps, \delta}$-DP mechanism $T_M \colon \bR^{n \times d} \to \bigotimes\limits_{j \in \brk{d^2}} I_j$ that satisfies the guarantee $\ex{X, T_M}{\llnorm{T_M\paren{X} - \paren{\eta_0 - m}}^2} \le \alpha^2$.
$T_M$ assumes the existence of an $\paren{\eps, \delta}$-DP mechanism $M \colon \bR^{n \times d} \to \bR^{d \times d}$ such that $\ex{X, M}{\mnorm{\Sigma}{M\paren{X} - \Sigma}^2} \le \frac{\alpha^2}{32}$.

\begin{algorithm}[htb]
    \caption{From Covariance Estimation to Natural Parameter Estimation}\label{alg:reduction}
    \hspace*{\algorithmicindent} \textbf{Input:} $X = \paren{X_1, \dots, X_n} \sim \cN\paren{0, \Sigma}^{\bigotimes n}$ and a mechanism $M \colon \bR^{n \times d} \to \bR^{d \times d}$. \\
    \hspace*{\algorithmicindent} \textbf{Output:} $T_M\paren{X} \in \bigotimes\limits_{j \in \brk{d^2}} I_j$.
    \begin{algorithmic}[1]
    \Procedure{\NatParamEst}{$X$}
        \State Let $\widehat{\Sigma} \coloneqq M\paren{X}$. \Comment{$\widehat{\Sigma}$ may not be equal to any possible $\Sigma$ generated by Algorithm~\ref{alg:sampling}.}
        \State Let $\widetilde{\Sigma}$ be the projection of $\widehat{\Sigma}$ onto the support of $\Sigma$ such that $\fnorm{\widehat{\Sigma} - \widetilde{\Sigma}}$ is minimized. \label{ln:proj}
        \State Let $\widetilde{U} \in \bR^{d \times d}$ be an upper-triangular matrix such that $\widetilde{\Sigma}^{- 1} = \widetilde{U} + \widetilde{U}^{\top}$.
        \State Let $T_M\paren{X} \coloneqq 2 \widetilde{U}^{\flat} - m$.
        \State \Return $T_M\paren{X}$.
    \EndProcedure
    \end{algorithmic}
\end{algorithm}

\begin{lemma}
\label{lem:reduction_cov}
Let $\Sigma \in \bR^{d \times d}$ be a covariance generated by Algorithm~\ref{alg:sampling}, and $X \sim \cN\paren{0, \Sigma}^{\bigotimes n}$ be a dataset.
If $M \colon \bR^{n \times d} \to \bR^{d \times d}$ is an $\paren{\eps, \delta}$-DP mechanism with $\ex{X, M}{\mnorm{\Sigma}{M\paren{X} - \Sigma}^2} \le \frac{\alpha^2}{32} \le \frac{1}{128}$, then Algorithm~\ref{alg:reduction} constructs an $\paren{\eps, \delta}$-DP mechanism $T_M \colon \bR^{n \times d} \to \bigotimes\limits_{j \in \brk{d^2}} I_j$ that satisfies $\ex{X, T_M}{\llnorm{T_M\paren{X} - \paren{\eta_0 - m}}^2} \le \alpha^2 \le \frac{1}{4}$.
\end{lemma}

\begin{proof}
The privacy guarantee of $T_M$ is an immediate consequence of the guarantee of $M$ and Lemma~\ref{lem:post_processing}.
Additionally, the only randomness used by $T_M$ is that used by $M$.
Having established that, we focus our attention on the accuracy guarantee.
By the definitions of $\eta_0$ and $T_M$, we have:
\begin{align}
    \ex{X, T_M}{\llnorm{T_M\paren{X} - \paren{\eta_0 - m}}^2} = 4 \ex{X, M}{\fnorm{\widetilde{U} - U}^2}. \label{eq:red1}
\end{align}
Additionally, the definitions of $U$ and $\widetilde{U}$ yield:
\begin{align}
    \fnorm{\widetilde{U} - U}^2 = \sum\limits_{j \geq i} \paren{\widetilde{U}_{i j} - U_{i j}}^2 &= \frac{1}{4} \sum\limits_{i \in \brk{d}} \brk{\paren{\widetilde{\Sigma}^{- 1}}_{i i} - \paren{\Sigma^{- 1}}_{i i}}^2 + \frac{1}{2} \sum\limits_{i \neq j} \brk{\paren{\widetilde{\Sigma}^{- 1}}_{i j} - \paren{\Sigma^{- 1}}_{i j}}^2 \nonumber \\
                                                                                                 &\le \frac{1}{2} \fnorm{\widetilde{\Sigma}^{- 1} - \Sigma^{- 1}}^2. \label{eq:red2}
\end{align}
Combining (\ref{eq:red1}) and (\ref{eq:red2}), we get:
\begin{align}
    \ex{X, T_M}{\llnorm{T_M\paren{X} - \paren{\eta_0 - m}}^2} \le 2 \ex{X, M}{\fnorm{\widetilde{\Sigma}^{-1} - \Sigma^{- 1}}^2}, \label{eq:red3}
\end{align}
Now, observe that:
\begin{align}
    \fnorm{\widetilde{\Sigma}^{- 1} - \Sigma^{- 1}} = \fnorm{\Sigma^{- \frac{1}{2}} \paren{\id - \Sigma^{\frac{1}{2}} \widetilde{\Sigma}^{- 1} \Sigma^{\frac{1}{2}}} \Sigma^{- \frac{1}{2}}} \overset{\paren{a}}&{\le} \llnorm{\Sigma^{- \frac{1}{2}}}^2 \fnorm{\id - \Sigma^{\frac{1}{2}} \widetilde{\Sigma}^{- 1} \Sigma^{\frac{1}{2}}} \nonumber \\
                                                                          \overset{\paren{b}}&{\le} \mnorm{\Sigma^{- 1}}{\widetilde{\Sigma}^{- 1} - \Sigma^{- 1}} \nonumber \\
                                                                                             &= \mnorm{\paren{\Sigma^{- \frac{1}{2}} \widetilde{\Sigma}^{- \frac{1}{2}}} \widetilde{\Sigma} \paren{\widetilde{\Sigma}^{- \frac{1}{2}} \Sigma^{- \frac{1}{2}}}}{\widetilde{\Sigma}^{- 1} - \Sigma^{- 1}} \nonumber \\
                                                                          \overset{\paren{c}}&{=} \mnorm{\widetilde{\Sigma}}{\paren{\widetilde{\Sigma}^{\frac{1}{2}} \Sigma^{\frac{1}{2}}} \paren{\widetilde{\Sigma}^{- 1} - \Sigma^{- 1}} \paren{\Sigma^{\frac{1}{2}} \widetilde{\Sigma}^{\frac{1}{2}}}} \nonumber \\
                                                                                             &= \mnorm{\widetilde{\Sigma}}{\widetilde{\Sigma} - \paren{\widetilde{\Sigma}^{\frac{1}{2}} \Sigma^{\frac{1}{2}}} \Sigma^{- 1} \paren{\Sigma^{\frac{1}{2}} \widetilde{\Sigma}^{\frac{1}{2}}}} \nonumber \\
                                                                          \overset{\paren{d}}&{=} \mnorm{\widetilde{\Sigma}}{\widetilde{\Sigma} - \Sigma} \nonumber \\
                                                                                             &= \fnorm{\widetilde{\Sigma}^{- \frac{1}{2}} \paren{\widetilde{\Sigma} - \Sigma} \widetilde{\Sigma}^{- \frac{1}{2}}} \nonumber \\
                                                                                             &\le \llnorm{\widetilde{\Sigma}^{- \frac{1}{2}}}^2 \fnorm{\widetilde{\Sigma} - \Sigma} \nonumber \\
                                                                          \overset{\paren{e}}&{\le} \fnorm{\widetilde{\Sigma} - \Sigma} \nonumber \\
                                                                                             &\le \fnorm{\widetilde{\Sigma} - \widehat{\Sigma}} + \fnorm{\widehat{\Sigma} - \Sigma} \nonumber \\
                                                                          \overset{\paren{f}}&{\le} 2 \fnorm{\widehat{\Sigma} - \Sigma} \nonumber \\
                                                                                             &= 2 \fnorm{\Sigma^{\frac{1}{2}} \paren{\Sigma^{- \frac{1}{2}} \widehat{\Sigma} \Sigma^{- \frac{1}{2}} - \id} \Sigma^{\frac{1}{2}}} \nonumber \\
                                                                                             &\le 2 \llnorm{\Sigma^{\frac{1}{2}}}^2 \fnorm{\Sigma^{- \frac{1}{2}} \widehat{\Sigma} \Sigma^{- \frac{1}{2}} - \id} \nonumber \\
                                                                                             &\le 4 \mnorm{\Sigma}{\widehat{\Sigma} - \Sigma}, \label{eq:red4}
\end{align}
where $\paren{a}$ uses Fact~\ref{fact:frob_product}, $\paren{b}$ uses the fact that $\Sigma^{- 1} \preceq \id$ (which holds by the first part of Lemma~\ref{lem:reduction_cov}) and that $\llnorm{\Sigma^{- \frac{1}{2}}} = \llnorm{\Sigma^{- 1}}^{\frac{1}{2}}$ for symmetric PSD matrices, $\paren{c}$ uses Lemma~\ref{lem:mah_rescale1} for $A = \Sigma^{- \frac{1}{2}} \widetilde{\Sigma}^{- \frac{1}{2}}$, $\paren{d}$ is by Lemma~\ref{lem:mah_rescale2}, $\paren{e}$ and $\paren{f}$ are thanks to the projection step in Line~\ref{ln:proj}, which ensures that $\widetilde{\Sigma}^{- 1} \preceq \id$ and $\fnorm{\widehat{\Sigma} - \widetilde{\Sigma}} \le \fnorm{\widehat{\Sigma} - \Sigma}$, and the last inequality uses that $\Sigma \preceq 2 \id$ (again by the first part of Lemma~\ref{lem:reduction_cov}).

Substituting based on (\ref{eq:red4}) into (\ref{eq:red3}), we get:
\[
    \ex{X, T_M}{\llnorm{T_M\paren{X} - \paren{\eta_0 - m}}^2} \le 32 \ex{X, M}{\mnorm{\Sigma}{M\paren{X} - \Sigma}^2},
\]
yielding the desired result.
\end{proof}
An immediate consequence of the above is the following:

\begin{corollary}
\label{cor:red_cov_lb}
Let $\Sigma \in \bR^{d \times d}$ be a covariance generated by Algorithm~\ref{alg:sampling}, and let $X \sim \cN\paren{0, \Sigma}^{\bigotimes n}$.
If any $\paren{\eps, \delta}$-DP mechanism $T \colon \bR^{n \times d} \to \bigotimes\limits_{j \in \brk{d^2}} I_j$ with $\ex{X, T}{\llnorm{T\paren{X} - \paren{\eta_0 - m}}^2} \le 32 \alpha^2 \le \frac{1}{4}$ requires at least $n \geq n_{\eta_0}$ samples, the same sample complexity lower bound holds for any $\paren{\eps, \delta}$-DP mechanism $M \colon \bR^{n \times d} \to \bR^{d \times d}$ that satisfies $\ex{X, M}{\mnorm{\Sigma}{M\paren{X} - \Sigma}^2} \le \alpha^2$.
\end{corollary}

Having established the reduction from estimating $\eta_0$ with respect to the $\ell_2$-norm to estimating $\Sigma$ with respect to the Mahalanobis norm, it remains now to apply Theorem~\ref{thm:lower_bound} to the former problem.
For that reason, we work towards bounding the various quantities involved in it.

First, as in Theorem~\ref{thm:lower_bound}, we assume the existence of an $\paren{\eps, \delta}$-DP mechanism $M \colon \bR^{n \times d} \to \bigotimes\limits_{j \in \brk{d^2}} I_j$ such that:
\[
    \ex{X, M}{\llnorm{M\paren{X} - \paren{\eta_0 - m}}^2} \le 32 \alpha^2 \le \frac{1}{96} \le \frac{\llnorm{R}^2}{24},
\]
where we appealed to the third part of Lemma~\ref{lem:sampling_analysis}.

The next step to establishing our lower bound is reasoning about the quantity $\ex{X_{\sim i}, M}{s^{\top} \Sigma_{T_0} s}$.
To do that, we identify an expression for $\Sigma_{T_0}$ for $\cN\paren{0, \Sigma}$ and use a result from~\cite{DiakonikolasKKLMS16}.
By Lemma~\ref{lem:exp_fam_norm_cov_alt}, we have $T_0 \equiv T \implies \Sigma_{T_0} = \Sigma_T$.
This yields:
\begin{align}
    \Sigma_T &= \ex{Y \sim \cN\paren{0, \Sigma}}{T\paren{Y} T\paren{Y}^{\top}} - \mu_T \mu_T^{\top} \nonumber \\
             &= \frac{1}{4} \ex{Y \sim \cN\paren{0, \Sigma}}{\paren{Y Y^{\top}}^{\flat} \paren{\paren{Y Y^{\top}}^{\flat}}^{\top}} - \frac{1}{4} \Sigma^{\flat} \paren{\Sigma^{\flat}}^{\top} \nonumber \\
             &= \frac{1}{4} \ex{Y \sim \cN\paren{0, \Sigma}}{\paren{Y \otimes Y} \paren{Y \otimes Y}^{\top}} - \frac{1}{4} \Sigma^{\flat} \paren{\Sigma^{\flat}}^{\top} \label{eq:cov_suff_stats}
\end{align}
To help us control the first term of (\ref{eq:cov_suff_stats}), we recall the following result:

\begin{lemma}
\label{lemma:fourth-order-symmetric}[Theorem $4.12$ from~\cite{DiakonikolasKKLMS16}]
Let $\cS_{\sym} = \brc{M^{\flat} \in \bR^{d^2} \colon M = M^{\top}}$ and $X \sim \cN\paren{0, \Sigma}$.
Let $M$ be the $d^2 \times d^2$ matrix given by $M \coloneqq \ex{}{\paren{X \otimes X} \paren{X \otimes X}^{\top}}$.
Then, as an operator on $\cS_{\sym}$, we have $M = 2 \Sigma^{\otimes 2} +  \paren{\Sigma^{\flat}} \paren{\Sigma^{\flat}}^{\top}$.
\end{lemma}

The above implies that, given any vector $v \in \bR^{d^2}$ that is the canonical flattening of a symmetric matrix, we have:
\[
    v^{\top} \ex{Y \sim \cN\paren{0, \Sigma}}{\paren{Y \otimes Y} \paren{Y \otimes Y}^{\top}} v = 2 v^{\top} \Sigma^{\otimes 2} v + v^{\top} \paren{\Sigma^{\flat}} \paren{\Sigma^{\flat}}^{\top} v.
\]
However, based on the definition given for $s$ in Section~\ref{sec:fing_exp_fams}, as well as Lemmas~\ref{lem:exp_fam_norm_cov_alt} and~\ref{lem:sampling_analysis}, $s^{\#}$ is upper-triangular, not symmetric.\footnote{$s$ would have been symmetric if we had written $\cN\paren{0, \Sigma}$ as an exponential family in the way suggested by Fact~\ref{fact:exp_fam_norm_cov}.
However, as we argued when motivating Lemma~\ref{lem:exp_fam_norm_cov_alt}, this parameterization cannot be used in our setting because Theorem~\ref{thm:lower_bound} necessitates that we generate the components of $\eta$ independently.}
Nevertheless, this does not cause a problem, since it is possible to prove the following generalization of the previous result:

\begin{lemma}
\label{lemma:fourth-order}
For $X \sim \cN\paren{0, \Sigma}$, $M$ be the $d^2 \times d^2$ matrix given by $M \coloneqq \ex{}{\paren{X \otimes X} \paren{X \otimes X}^{\top}}$.
Then, we have:
\[
    v^{\top} M v = v^{\top} \Sigma^{\otimes 2} v + v^{\top} \Sigma^{\otimes 2} \brk{\paren{v^{\#}}^{\top}}^{\flat} + v^{\top} \paren{\Sigma^{\flat}} \paren{\Sigma^{\flat}}^{\top} v, \forall v \in \bR^{d^2}.
\]
\end{lemma}
In the above, the term $2 v^{\top} \Sigma^{\otimes 2} v$ that appeared in Lemma~\ref{lemma:fourth-order-symmetric} is replaced by the sum $v^{\top} \Sigma^{\otimes 2} v + v^{\top} \Sigma^{\otimes 2} \brk{\paren{v^{\#}}^{\top}}^{\flat}$.
The second term of this sum is the bilinear form $v^{\top} \Sigma^{\otimes 2} \brk{\paren{v^{\#}}^{\top}}^{\flat}$.
Observe that $\brk{\paren{v^{\#}}^{\top}}^{\flat}$ is the vector one gets by taking the transpose of the matrix $v^{\#}$ and then flattening the resulting matrix.
The proof of this result follows the same basic steps as that of Lemma~\ref{lemma:fourth-order-symmetric}, with the only difference essentially being that we need to use the SVD of the matrix $v^{\#}$.
For that reason, we do not repeat the proof here and point readers to~\cite{DiakonikolasKKLMS16} for the full argument.

Using the above, we prove the following lemma:

\begin{lemma}
\label{lem:suff_stats_cov_ub}
Let $\Sigma \in \bR^{d \times d}$ be a covariance matrix generated by Algorithm~\ref{alg:sampling}, and let $X \sim \cN\paren{0, \Sigma}^{\bigotimes n}$ be a dataset.
Also, let $M \colon \bR^{n \times d} \to \bigotimes\limits_{j \in \brk{d^2}} I_j$ be an $\paren{\eps, \delta}$-DP mechanism such that $\ex{X, M}{\llnorm{M\paren{X} - \paren{\eta_0 - m}}^2} \le 32 \alpha^2$.
Then, we have $\ex{X_{\sim i}, M}{s^{\top} \Sigma_T s} \le \frac{4 \alpha^2}{d^4}$.
\end{lemma}

\begin{proof}
Let $v \coloneqq \brk{\paren{s^{\#}}^{\top}}^{\flat}$.
By Lemma~\ref{lemma:fourth-order} and (\ref{eq:cov_suff_stats}), we get that:
\begin{align}
    \ex{X_{\sim i}, M}{s^{\top} \Sigma_T s} = \frac{1}{4} \ex{X_{\sim i}, M}{s^{\top} \Sigma^{\otimes 2} s} + \frac{1}{4} \ex{X_{\sim i}, M}{s^{\top} \Sigma^{\otimes 2} v}. \label{eq:suff_stats_cov_ub1}
\end{align}
We argue that both terms of (\ref{eq:suff_stats_cov_ub1}) are upper-bounded by the same quantity.
For the first term, given that $\Sigma^{\otimes 2}$ is symmetric, the definition of the spectral norm yields:
\begin{align}
    \ex{X_{\sim i}, M}{s^{\top} \Sigma^{\otimes 2} s} = \ex{X_{\sim i}, M}{\llnorm{s}^2 \frac{s^{\top}}{\llnorm{s}} \Sigma^{\otimes 2} \frac{s}{\llnorm{s}}} \le \llnorm{\Sigma^{\otimes 2}} \ex{X_{\sim i}, M}{\llnorm{s}^2}. \label{eq:suff_stats_cov_ub2}
\end{align}
Now, for the second term, the first version of Fact~\ref{fact:cs} for $x = s, y = \Sigma^{\otimes 2} v$, as well as the definition of the spectral norm, yield:
\begin{align}
    \ex{X_{\sim i}, M}{s^{\top} \Sigma^{\otimes 2} v} = \ex{X_{\sim i}, M}{\iprod{s, \Sigma^{\otimes 2} v}} \le \ex{X_{\sim i}, M}{\llnorm{s} \llnorm{\Sigma^{\otimes 2} v}} &= \ex{X_{\sim i}, M}{\llnorm{s} \llnorm{v} \llnorm{\Sigma^{\otimes 2} \frac{v}{\llnorm{v}}}} \nonumber \\
     &\le \llnorm{\Sigma^{\otimes 2}} \ex{X_{\sim i}, M}{\llnorm{s} \llnorm{v}} \nonumber \\
     &= \llnorm{\Sigma^{\otimes 2}} \ex{X_{\sim i}, M}{\llnorm{s}^2}, \label{eq:suff_stats_cov_ub3}
\end{align}
where the last equality uses the observation that $\llnorm{v} = \llnorm{\brk{\paren{s^{\#}}^{\top}}^{\flat}} = \fnorm{\paren{s^{\#}}^{\top}} = \fnorm{s^{\#}} = \llnorm{s}$.

Substituting to (\ref{eq:suff_stats_cov_ub1}) based on (\ref{eq:suff_stats_cov_ub2}) and (\ref{eq:suff_stats_cov_ub3}), we get:
\begin{align*}
    \ex{X_{\sim i}, M}{s^{\top} \Sigma_T s} &\le \frac{1}{2} \llnorm{\Sigma^{\otimes 2}} \ex{X_{\sim i}, M}{\llnorm{s}^2} \\
                         \overset{\paren{a}}&{=} \frac{1}{2} \llnorm{\Sigma}^2 \ex{X_{\sim i}, M}{\llnorm{s}^2} \\
                         \overset{\paren{b}}&{\le} 2 \ex{X_{\sim i}, M}{\sum\limits_{j \in \brk{d^2}} \brk{\frac{R_j^2}{4} - \paren{\eta_{0, j} - m_j}^2}^2 \brk{M_j\paren{X_{\sim i}} - \paren{\eta_{0, j} - m_j}}^2} \\
                                            &\le \frac{\infnorm{R}^4}{16} \cdot 2 \ex{X_{\sim i}, M}{\llnorm{M\paren{X_{\sim i}} - \paren{\eta_0 - m}}^2} \\
                         \overset{\paren{c}}&{\le} \frac{\infnorm{R}^4}{8} \cdot 32 \alpha^2 \\
                                            &= 4 \infnorm{R}^4 \alpha^2 \\
                                            &= \frac{4 \alpha^2}{d^4},
\end{align*}
where $\paren{a}$ is by Fact~\ref{fact:kronecker_norm}, $\paren{b}$ is by the first part of Lemma~\ref{lem:sampling_analysis}, $\paren{c}$ is by the assumption that $\ex{X, M}{\llnorm{M\paren{X} - \paren{\eta_0 - m}}^2} \le 32 \alpha^2$, and the last equality is by the third part of Lemma~\ref{lem:sampling_analysis}.
\end{proof}

Now, it remains to reason about the term involving the tail probabilities of $\llnorm{T\paren{X_i} - \mu_T}$.
Our approach closely follows  Remark~\ref{rem:exp_chernoff}.
The following lemma is devoted to implementing this.

\begin{lemma}
\label{lem:tail_prob_cov_ub1}
Let $\Sigma \in \bR^{d \times d}$ be a covariance matrix generated by Algorithm~\ref{alg:sampling}, and $X_i \sim \cN\paren{0, \Sigma}$.
For $T \geq \frac{1}{3 d^2}$, we have:
\[
    \int\limits_T^{\infty} \pr{X_i}{\llnorm{T\paren{X_i} - \mu_T} > 4 d^2 t} \, dt \le \frac{2}{3 c d^2} e^{- d^2 \paren{3 c T - 2 \ln\paren{3}}},
\]
where $c \approx 0.036425$.
\end{lemma}

\begin{proof}
Setting $Y \coloneqq \Sigma^{- \frac{1}{2}} X_i \sim \cN\paren{0, \id}$, we observe that:
\begin{align*}
    \llnorm{T\paren{X_i} - \mu_T} = \frac{1}{2} \llnorm{\paren{X_i X_i^{\top}}^{\flat} - \Sigma^{\flat}} &= \frac{1}{2} \fnorm{X_i X_i^{\top} - \Sigma} \\
                                                                                                         &= \frac{1}{2} \fnorm{\Sigma^{\frac{1}{2}} \brk{\paren{\Sigma^{- \frac{1}{2}} X_i} \paren{\Sigma^{- \frac{1}{2}} X_i}^{\top} - \id} \Sigma^{\frac{1}{2}}} \\
                                                                                      \overset{\paren{a}}&{\le} \frac{\llnorm{\Sigma^{\frac{1}{2}}}^2}{2} \fnorm{Y Y^{\top} - \id} \\
                                                                                      \overset{\paren{b}}&\le \fnorm{Y Y^{\top} - \ex{}{Y Y^{\top}}},
\end{align*}
where $\paren{a}$ uses Fact~\ref{fact:frob_product}, and $\paren{b}$ uses the assumption that $\Sigma \preceq 2 \id$ and the remark that, for symmetric PSD matrices, $\llnorm{\Sigma^{\frac{1}{2}}} = \sqrt{\llnorm{\Sigma}}$.

Thus, we have:
\[
    \int\limits_T^{\infty} \pr{X_i}{\llnorm{T\paren{X_i} - \mu_T} > 4 d^2 t} \, dt \le \int\limits_T^{\infty} \pr{}{\fnorm{Y Y^{\top} - \ex{}{Y Y^{\top}}} > 4 d^2 t} \, dt.
\]
Working as we described in Remark~\ref{rem:exp_chernoff}, we get:
\begin{align}
    \int\limits_T^{\infty} \pr{X_i}{\llnorm{T\paren{X_i} - \mu_T} > 4 d^2 t} \, dt &\le \int\limits_T^{\infty} \pr{}{\llnorm{\paren{Y Y^{\top} - \ex{}{Y Y^{\top}}}^{\flat}} > 4 d^2 t} \, dt \nonumber \\
                                                                                   &\le 9^{d^2} \int\limits_T^{\infty} \pr{}{\abs{\iprod{M^{\flat}, \paren{Y Y^{\top} - \ex{}{Y Y^{\top}}}^{\flat}}} > 3 d^2 t} \, dt \nonumber \\
                                                                                   &= 9^{d^2} \int\limits_T^{\infty} \pr{}{\abs{\iprod{M, Y Y^{\top} - \ex{}{Y Y^{\top}}}} > 3 d^2 t} \, dt \nonumber \\
                                                                                   &= 9^{d^2} \int\limits_T^{\infty} \pr{}{\abs{Y^{\top} M Y - \ex{}{Y^{\top} M Y}} > 3 d^2 t} \, dt \nonumber \\
                                                                                   &= 9^{d^2} \int\limits_T^{\infty} \pr{}{\abs{Y^{\top} \frac{M + M^{\top}}{2} Y - \ex{}{Y^{\top} \frac{M + M^{\top}}{2} Y}} > 3 d^2 t} \, dt, \label{eq:tail_prob_cov_equiv1}
\end{align}
where $M$ is a matrix with $\fnorm{M} = 1$.

We have that $Y \sim \cN\paren{0, \id}$, and that $\frac{M + M^{\top}}{2}$ is symmetric, so we can apply Fact~\ref{fact:hanson-wright} to (\ref{eq:tail_prob_cov_equiv1}).
Since $\fnorm{\frac{M + M^{\top}}{2}} \le \frac{\fnorm{M} + \fnorm{M^{\top}}}{2} = 1$ and $\llnorm{\frac{M + M^{\top}}{2}} \le \fnorm{\frac{M + M^{\top}}{2}} \le 1$, we get:
\[
    \int\limits_T^{\infty} \pr{X_i}{\llnorm{T\paren{X_i} - \mu_T} > 4 d^2 t} \, dt \le 2 \cdot 9^{d^2} \int\limits_T^{\infty} \exp\paren{- c \min\brc{9 d^4 t^2, 3 d^2 t}} \, dt.
\]
By our assumption that on $T$, we get $\min\brc{9 d^4 t^2, 3 d^2 t} = 3 d^2 t, \forall t \geq T$.
This leads to:
\[
    \int\limits_T^{\infty} \pr{X_i}{\llnorm{T\paren{X_i} - \mu_T} > 4 d^2 t} \, dt \le 2 \cdot 9^{d^2} \int\limits_T^{\infty} \exp\paren{- 3 c d^2 t} \, dt = \frac{2 \cdot 9^{d^2}}{3 c d^2} e^{- 3 c d^2 T} = \frac{2}{3 c d^2} e^{- d^2 \paren{3 c T - 2 \ln\paren{3}}}.
\]
\end{proof}

We are now ready to prove the main result of this section.

\begin{proof}[Proof of Theorem~\ref{thm:lb-maha}]
We assume that the process generating $\Sigma$ is that of Algorithm~\ref{alg:sampling}.
By the first part of Lemma~\ref{lem:sampling_analysis}, the condition $\id \preceq \Sigma \preceq 2 \id$ is satisfied.
We assume that there exists an $\paren{\eps, \delta}$-DP mechanism $M \colon \bR^{n \times d} \to \bigotimes\limits_{j \in \brk{d^2}} I_j$ that satisfies $\ex{X, M}{\llnorm{M\paren{X} - \paren{\eta_0 - m}}^2} \le 32 \alpha^2 \le \frac{1}{96}$.
Combining the results of Lemmas~\ref{lem:suff_stats_cov_ub}, and~\ref{lem:tail_prob_cov_ub1}, with Theorem~\ref{thm:lower_bound}, we get:
\begin{align}
    n \paren{2 \delta T + 4 \frac{\alpha \eps}{d^2} + \frac{4}{3 c d^2} e^{- d^2 \paren{3 c T - 2 \ln\paren{3}}}} \geq \frac{1}{48} \paren{1 - \frac{1}{2 d}} \geq \frac{1}{96}, \label{eq:ineq1}
\end{align}
where $T$ must be at least $\frac{1}{3 d^2}$ and $c \approx 0.036425$ (by Lemma~\ref{lem:tail_prob_cov_ub1}).

It remains to set the values of $T$ and $\delta$ appropriately so that:
\[
    \delta T \geq \frac{4}{3 c d^2} e^{- d^2 \paren{3 c T - 2 \ln\paren{3}}} \text{ and } 3 n \delta T \le \frac{1}{192}.
\]
The first of these two conditions yields:
\[
    \delta T \geq \frac{4}{3 c d^2} e^{- d^2 \paren{3 c T - 2 \ln\paren{3}}} \iff T e^{d^2 \paren{3 c T - 2 \ln\paren{3}}} \geq \frac{4}{3 c d^2} \cdot \frac{1}{\delta}.
\]
We set $T = \frac{1}{3 c} \paren{2 \ln\paren{3} + \frac{1}{d^2} \ln\paren{\frac{1}{\delta}}}$.
Assuming that $\frac{1}{c} \ln\paren{\frac{1}{\delta}} \geq 1 \iff \delta \le e^{- c} \approx 0.96423040846$, this satisfies the constraint of Lemma~\ref{lem:tail_prob_cov_ub1}.
Then, the above becomes:
\[
    \frac{1}{3 c} \paren{2 \ln\paren{3} + \frac{1}{d^2} \ln\paren{\frac{1}{\delta}}} \frac{1}{\delta} \geq \frac{1}{c d^2} \cdot \frac{1}{\delta} \iff \ln\paren{\frac{1}{\delta}} \geq 3 - 2 \ln\paren{3} d^2.
\]
The RHS of the last inequality is $< 0$ for $d \geq 2$.
For $d = 1$, we get the constraint $\delta \le e^{2 \ln\paren{3} - 3} \approx 0.44$, which is stricter than $\delta \le e^{- c}$.
Respecting this constraint, we proceed to identify a range of values for $\delta$ such that $3 n \delta T \le \frac{1}{192}$.
This is equivalent to:
\begin{align}
    \frac{\delta}{c} \paren{2 \ln\paren{3} + \frac{1}{d^2} \ln\paren{\frac{1}{\delta}}} \le \frac{1}{192 n} \iff \delta \paren{1 + \frac{\ln\paren{\frac{1}{\delta}}}{2 \ln\paren{3} d^2}} \le \frac{c}{384 \ln\paren{3} n}. \label{eq:ineq2}
\end{align}
The above inequality yields a constraint on $\delta$.
To ensure that $\ln\paren{\frac{1}{\delta}} \geq 0$, we need $\frac{c}{384 \ln\paren{3} n \delta} - 1 \geq 0 \iff \delta \le \frac{c}{384 \ln\paren{3} n}$.
We know from Lemma~\ref{lem:tail_prob_cov_ub1} that $c \approx 0.036425$.
Thus, this last constraint is stricter than $\delta \le 0.44$ for all $n \geq 1$.

For (\ref{eq:ineq2}) to hold, we will show that it suffices to have:
\[
    \delta \le \min\brc{\frac{c}{768 \ln\paren{3} n}, \frac{c d^2}{768 n \ln\paren{\frac{384 n}{c d^2}}}},
\]
which trivially satisfies the previous constraint.

We set $\phi \coloneqq \frac{c d^2}{384 n}$, so the proposed range for $\delta$ becomes $\delta \le \min\brc{\frac{\phi}{2 \ln\paren{3} d^2}, \frac{\phi}{2 \ln\paren{\frac{1}{\phi}}}}$.
Additionally, (\ref{eq:ineq2}) can be equivalently written as:
\begin{align}
    \delta \paren{1 + \frac{\ln\paren{\frac{1}{\delta}}}{2 \ln\paren{3} d^2}} \le \frac{\phi}{\ln\paren{3} d^2}. \label{eq:ineq3}
\end{align}
Observe that, depending on how $\delta$ compares with $e^{- 2 \ln\paren{3} d^2}$ determines which of the terms $1$ and $\frac{\ln\paren{\frac{1}{\delta}}}{2 \ln\paren{3} d^2}$ dominates.
Thus, our strategy to verify our claim is by considering cases based on whether the proposed range for $\delta$ includes $e^{- 2 \ln\paren{3} d^2}$.
This yields:
\begin{enumerate}
    \item $\min\brc{\frac{\phi}{2 \ln\paren{3} d^2}, \frac{\phi}{2 \ln\paren{\frac{1}{\phi}}}} < e^{- 2 \ln\paren{3} d^2}$.
    We start by arguing that:
    \begin{align}
        \min\brc{\frac{\phi}{2 \ln\paren{3} d^2}, \frac{\phi}{2 \ln\paren{\frac{1}{\phi}}}} < e^{- 2 \ln\paren{3} d^2} \iff \frac{\phi}{2 \ln\paren{\frac{1}{\phi}}} < e^{- 2 \ln\paren{3} d^2}. \label{eq:equivalence_cov}
    \end{align}
    The $\impliedby$ direction is trivial, so we focus on the $\implies$ direction.
    
    We assume that $\min\brc{\frac{\phi}{2 \ln\paren{3} d^2}, \frac{\phi}{2 \ln\paren{\frac{1}{\phi}}}} < e^{- 2 \ln\paren{3} d^2}$ but $\frac{\phi}{2 \ln\paren{\frac{1}{\phi}}} \geq e^{- 2 \ln\paren{3} d^2}$.
    Then, it must be the case that $\frac{\phi}{2 \ln\paren{3} d^2} < e^{- 2 \ln\paren{3} d^2}$.
    Setting $y \coloneqq \frac{1}{\phi}$, the above system of inequalities can be written in the form:
    \[
        \begin{cases}
            \frac{e^{2 \ln\paren{3} d^2}}{2} \geq y \ln\paren{y} \\
            y > \frac{e^{2 \ln\paren{3} d^2}}{2 \ln\paren{3} d^2}
        \end{cases}.
    \]
    The function $y \ln\paren{y}$ is increasing for $y \geq e^{- 1}$.
    We have $\frac{e^{2 \ln\paren{3} d^2}}{2 \ln\paren{3} d^2} \geq e^{- 1}, \forall d \geq 1$.
    Thus, for $y > \frac{e^{2 \ln\paren{3} d^2}}{2 \ln\paren{3} d^2}$, we get:
    \[
        \frac{e^{2 \ln\paren{3} d^2}}{2} \geq y \ln\paren{y} > \frac{e^{2 \ln\paren{3} d^2}}{2 \ln\paren{3} d^2} \ln\paren{\frac{e^{2 \ln\paren{3} d^2}}{2 \ln\paren{3} d^2}} \implies \ln\paren{3} d^2 < \ln\paren{2 \ln\paren{3} d^2}.
    \]
    However, the above inequality cannot be satisfied, which leads to a contradiction.
    
    Consequently, we have shown (\ref{eq:equivalence_cov}).
    Since we have assumed that $\min\brc{\frac{\phi}{2 \ln\paren{3} d^2}, \frac{\phi}{2 \ln\paren{\frac{1}{\phi}}}} < e^{- 2 \ln\paren{3} d^2}$, it must be the case that $\frac{\phi}{2 \ln\paren{\frac{1}{\phi}}} < e^{- 2 \ln\paren{3} d^2}$.
    
    We now turn our attention again to (\ref{eq:ineq3}).
    We remark that, for $\delta < e^{- 2 \ln\paren{3} d^2}$, we get $\frac{\ln\paren{\frac{1}{\delta}}}{2 \ln\paren{3} d^2} > 1$.
    As a result, in order to satisfy (\ref{eq:ineq3}), it suffices to have:
    \begin{align}
        2 \delta \frac{\ln\paren{\frac{1}{\delta}}}{2 \ln\paren{3} d^2} \le \frac{\phi}{\ln\paren{3} d^2} \iff \delta \ln\paren{\frac{1}{\delta}} \le \phi. \label{eq:ineq4}
    \end{align}
    Observe that the function $\delta \ln\paren{\frac{1}{\delta}}$ is increasing for $\delta < e^{- 2 \ln\paren{3} d^2} < e^{- 1}$.
    Combined with our previous remarks, this implies that, in order to verify (\ref{eq:ineq4}) for the values of $\delta$ we picked, it suffices to show that it holds for $\delta = \frac{\phi}{2 \ln\paren{\frac{1}{\phi}}}$.
    So, we have to verify the inequality:
    \[
        \paren{\frac{\phi}{2 \ln\paren{\frac{1}{\phi}}}} \ln\paren{\frac{2 \ln\paren{\frac{1}{\phi}}}{\phi}} \le \phi \iff \frac{\ln\paren{2 \ln\paren{\frac{1}{\phi}}}}{\ln\paren{\frac{1}{\phi}}} \le 1.
    \]
    Setting $z \coloneqq \ln\paren{\frac{1}{\phi}}$, this becomes equivalent to $\frac{\ln\paren{2 z}}{z} < 1$, which holds for all $z > 0 \iff \ln\paren{\frac{1}{\phi}} > 0 \iff \phi < 1$.
    This last condition is satisfied by all $\phi$ such that $\frac{\phi}{2 \ln\paren{\frac{1}{\phi}}} < e^{- 2 \ln\paren{3} d^2}$, so the desired result has been established.
    \item $\min\brc{\frac{\phi}{2 \ln\paren{3} d^2}, \frac{\phi}{2 \ln\paren{\frac{1}{\phi}}}} \geq e^{- 2 \ln\paren{3} d^2}$.
    For this case, we consider two sub-cases, depending on how $\delta$ compares with $e^{- 2 \ln\paren{3} d^2}$.
    We have:
    \begin{itemize}
        \item $\delta < e^{- 2 \ln\paren{3} d^2}$.
        As before, we argue that it suffices to have $\delta \ln\paren{\frac{1}{\delta}} \le \phi$.
        Since $\delta \ln\paren{\frac{1}{\delta}}$ is increasing for $\delta < e^{- 2 \ln\paren{3} d^2} < e^{- 1}$, all we have to do is show that $e^{- 2 \ln\paren{3} d^2} 2 \ln\paren{3} d^2 \le \phi$.
        This follows by our assumption that $\min\brc{\frac{\phi}{2 \ln\paren{3} d^2}, \frac{\phi}{2 \ln\paren{\frac{1}{\phi}}}} \geq e^{- 2 \ln\paren{3} d^2}$.
        \item $\delta \geq e^{- 2 \ln\paren{3} d^2}$.
        This implies $\frac{\ln\paren{\frac{1}{\delta}}}{2 \ln\paren{3} d^2} \le 1$.
        Thus, for (\ref{eq:ineq3}) to hold, it suffices that:
        \[
            2 \delta \le \frac{\phi}{\ln\paren{3} d^2} \iff \delta \le \frac{\phi}{2 \ln\paren{3} d^2}.
        \]
        This is satisfied, because our proposed range for $\delta$ is $\delta \le \min\brc{\frac{\phi}{2 \ln\paren{3} d^2}, \frac{\phi}{2 \ln\paren{\frac{1}{\phi}}}}$.
    \end{itemize}
\end{enumerate}
We have established that our proposed values of $\delta$ and $T$ imply:
\[
    \delta T \geq \frac{4}{3 c d^2} e^{- d^2 \paren{3 c T - 2 \ln\paren{3}}} \text{ and } 3 n \delta T \le \frac{1}{192},
\]
while respecting all the constraints.

We substitute this to (\ref{eq:ineq1}) and get $n \geq \frac{d^2}{768 \alpha \eps}$.
Appealing to Corollary~\ref{cor:red_cov_lb} completes the proof.
\end{proof}

\subsection{Estimation with Respect to the Spectral Norm}
\label{subsec:spec_norm_lb}

Here, we prove a lower bound for covariance estimation of high-dimensional Gaussians in spectral norm under the constraint of approximate DP.
Employing a reduction-based approach from Mahalanobis estimation to spectral estimation, we directly leverage the results of Section~\ref{subsec:mah_norm_lb}.

\begin{theorem}[Gaussian Covariance Estimation in Spectral Norm]\label{thm:lb-spectral}
There exists a distribution $\cD$ over covariance matrices $\Sigma \in \bR^{d \times d}$ with $\id \preceq \Sigma \preceq 2 \id$ such that, given $\Sigma \sim \cD$ and $X \sim \cN\paren{0, \Sigma}^{\bigotimes n}$, for any $\alpha = \cO\paren{\frac{1}{\sqrt{d}}}$ and any $\paren{\eps, \delta}$-DP mechanism $M \colon \bR^{n \times d} \to \bR^{d \times d}$ with $\eps, \delta \in \brk{0, 1}$, and $\delta \le \cO\paren{\min\brc{\frac{1}{n}, \frac{d^2}{n \log\paren{\frac{n}{d^2}}}}}$ that satisfies $\ex{X, M}{\llnorm{\Sigma^{- \frac{1}{2}} \paren{M\paren{X} - \Sigma} \Sigma^{- \frac{1}{2}}}^2} \le \alpha^2$, it must hold that $n \geq \Omega\paren{\frac{d^{1.5}}{\alpha \eps}}$.
\end{theorem}

\begin{proof}
First, assume that $\Sigma$ is generated by Algorithm~\ref{alg:sampling} (as was assumed in all of Section~\ref{subsec:mah_norm_lb}).
Additionally, assume that there exists an $\paren{\eps, \delta}$-DP mechanism $M \colon \bR^{n \times d} \to \bR^{d \times d}$ with $\eps \in \brk{0, 1}$ and $\delta \le \cO\paren{\min\brc{\frac{1}{n}, \frac{d^2}{n \log\paren{\frac{n}{d^2}}}}}$ that satisfies $\ex{X, M}{\llnorm{\Sigma^{- \frac{1}{2}} \paren{M\paren{X} - \Sigma} \Sigma^{- \frac{1}{2}}}^2} \le \alpha^2$ and uses $n = o\paren{\frac{d^{1.5}}{\alpha \eps}}$ samples.
Using the standard matrix-norm inequality $\fnorm{A} \le \sqrt{d} \llnorm{A}, \forall A \in \bR^{k \times \ell}$, we get:
\begin{align*}
    \ex{X, M}{\mnorm{\Sigma}{M\paren{X} - \Sigma}^2} = \ex{X, M}{\fnorm{\Sigma^{- \frac{1}{2}} \paren{M\paren{X} - \Sigma} \Sigma^{- \frac{1}{2}}}^2} &\le d \ex{X, M}{\llnorm{\Sigma^{- \frac{1}{2}} \paren{M\paren{X} - \Sigma} \Sigma^{- \frac{1}{2}}}^2} \\
                                                                                                                                                      &\le d \alpha^2.
\end{align*}
By Theorem~\ref{thm:lb-maha} for $\alpha \to \sqrt{d} \alpha$, we get that any mechanism with the characteristics of $M$ that estimates a Gaussian covariance matrix in Mahalanobis error $\sqrt{d} \alpha$ requires $n \geq \Omega\paren{\frac{d^{1.5}}{\alpha \eps}}$ samples.
This leads to a contradiction, implying the desired sample complexity lower bound.
\end{proof}

\section{Mean Estimation of Heavy-Tailed Distributions}
\label{sec:heavy-tailed}
    
In this section, we prove a lower bound for mean estimation of high-dimensional distributions with second moments bounded by $1$.
For this, we use DP Assouad's method from~\cite{AcharyaSZ21} (Lemma~\ref{lem:dp-assouad}).

\begin{theorem}[Heavy-Tailed Mean Estimation]\label{thm:lb-second-moment}
Let $\cD$ be an arbitrary distribution over $\bR^d$ with second moments bounded by $1$ and unknown mean $\mu$.
Then, for any $\alpha \le 1$, any $\paren{\eps, \delta}$-DP mechanism with $\delta \le \eps$ that takes $X \sim \cD^{\bigotimes n}$ as input and outputs $M\paren{X}$ satisfying $\ex{X, M}{\llnorm{M\paren{X} - \mu}^2} \le \alpha^2$ requires $n \geq \Omega\paren{\frac{d}{\alpha^2 \eps}}$.
\end{theorem}

\begin{proof}
We define the following family of distributions, such that the second moment of each distribution is upper-bounded by $1$.
The loss function between the true distribution and the estimate is the squared $\ell_2$-distance between their means.
For each $v \in \cE_d = \brc{\pm 1}^d$, we define the distribution $\cD_v$ over $\bR^d$ as follows.
Let $t > 0$ and $0 < p < 1$.
For any $X \sim \cD_v$, we have:
\[
    X_i =
    \begin{cases}
        v_i t, & \text{with probability } p \\
        0    , & \text{with probability } 1-p
    \end{cases}
    , \forall i \in \brk{d}.
\]
Then $\ex{X \sim \cD_v}{X} = p t v$.
We show that the second moment of this distribution is upper-bounded by $p t^2$, that is, $\ex{X \sim \cD_v}{\iprod{X - p t v, u}^2} \le p t^2, \forall u \in \bS^{d-1}$.
Note that, for any $i \in \brk{d}$, we have $\ex{}{\paren{X_i - p t v_i}^2} = p \paren{1 - p} t^2$.
We omit the subscript of the expectation when the context is clear.
We have:
\begin{align*}
    \ex{}{\iprod{X - p t v, u}^2} &= \ex{}{\paren{\sum\limits_{i = 1}^d \paren{X_i - p t v_i} u_i}^2} \\
                                  &= \sum\limits_{i = 1}^d {\ex{}{\paren{X_i - p t v_i}^2} u_i^2} + \sum\limits_{i \neq j}{\ex{}{X_i - p t v_i} \ex{}{X_j - p t v_j} u_i u_j} \\
                                  &= \sum\limits_{i = 1}^d p \paren{1 - p} t^2 u_i^2 \\
                                  &= p \paren{1 - p} t^2 \\
                                  &\le p t^2.
\end{align*}
We want this to be upper-bounded by $1$:
\begin{align}
    p t^2 = 1. \label{eq:bounded-moment}
\end{align}
    
Next, for any $u, v \in \cE_d$, we bound the loss function between $\cD_{v_1}$ and $\cD_{v_2}$.
As mentioned before, we define $\ell\paren{\theta\paren{\cD_u},\theta\paren{\cD_v}} = \llnorm{\theta\paren{\cD_u} - \theta\paren{\cD_v}}^2$.
We have the following:
\begin{align*}
    \llnorm{\theta\paren{\cD_u} - \theta\paren{\cD_v}}^2 = \sum\limits_{i = 1}^d p^2 t^2 \paren{v_i - u_i}^2 = p^2 t^2 \sum\limits_{i \colon u_i \neq v_i} \paren{v_i - u_i}^2 = 4 p^2 t^2 \sum\limits_{i = 1}^d \mathds{1}\brc{v_i \neq u_i}.
\end{align*}
Using the notation in Lemma~\ref{lem:dp-assouad}, we have $\tau = 2 p^2 t^2$.
Using the same lemma, we get:
\begin{align*}
    R\paren{\cP, \ell, \eps, \delta} \geq \frac{d \tau}{2} \paren{0.9 e^{- 10 \eps D} - 10 \delta D} \geq \frac{d \tau}{2} \brk{0.9 \paren{1 - 10 \eps D} - 10 \delta D} &= \frac{d \tau}{2} \brk{0.9 - D \paren{9 \eps + 10 \delta}} \\
    &\geq \frac{d \tau}{2} \brk{0.9 - 10 D \paren{\eps + \delta}}.
\end{align*}
Setting $R\paren{\cP, \ell, \eps, \delta} \le \alpha^2$ and rearranging the above, we get:
\begin{align}
    D \geq \frac{1}{10 \paren{\eps + \delta}} \cdot \paren{0.9 - \frac{2 \alpha^2}{d \tau}}. \label{eq:second-moment-assouad}
\end{align}
Now, from (\ref{eq:bounded-moment}), $\tau = 2 p^2 t^2 = 2p$.
Using this and (\ref{eq:second-moment-assouad}), and setting $p = \frac{2 \alpha^2}{d}$, we have the following:
\begin{align}
    D \geq \frac{0.04}{\eps + \delta}. \label{eq:second-moment-assouad-2}
\end{align}
For $i \in \brk{d}$, we define the mixtures $p_{+i}$ and $p_{-i}$ as in Lemma~\ref{lem:dp-assouad}.
We define a coupling $\paren{X, Y}$ between $p_{+ i}$ and $p_{- i}$ as follows.
For every row $X_j$ in $X$, if $X_j^i \neq 0$, then $Y_j^i = - X_j^i$, otherwise $Y_j^i = X_j^i$.
For all other coordinates $k \neq i$, $Y_j^k = X_j^k$.
Notice that, for all $k \neq i$, the distributions for coordinate $k$ for both $p_{+ i}$ and $p_{- i}$ are identical because all coordinates are independent.
Therefore, setting $Y_j^k = X_j^k$, these coordinates of $X_j$ and $Y_j$ have identical distributions for all $j \in \brk{n}$.
The coordinate $i$ of $X_j$ is $t$ with probability $p$, and $0$ with probability $1 - p$.
Hence, the coordinate $i$ of $Y_j$ is $- t$ with probability $p$, and $0$ with probability $1 - p$.
This is true for all $j \in \brk{n}$.
Since all $n$ rows in both $X$ and $Y$ are chosen independently at random, the distribution of $X$ is $p_{+ i}$ and the distribution of $Y$ is $p_{- i}$.
Therefore, $\paren{X, Y}$ is a valid coupling of $p_{+ i}$ and $p_{- i}$.
It remains to determine the value of $D = \ex{}{\dHam\paren{X, Y}}$.
Note that, for each $j \in \brk{n}$, $X_j$ can only differ from $Y_j$ when $X_j^i \neq 0$.
This can only happen with probability $p$.
Thus, by linearity of expectation, $D = n p$.
    
Finally, using the fact that $p = \frac{2 \alpha^2}{d}$ and $t= \frac{\sqrt{d}}{\sqrt{2} \alpha}$, and substituting this in (\ref{eq:second-moment-assouad-2}), we have:
\[
    n p \geq \frac{0.04}{\eps + \delta} \implies n \geq \frac{d}{50 \paren{\eps + \delta} \alpha^2}.
\]
When $\eps \geq \delta$, this gives us the desired lower bound.
\end{proof}

\section{Conclusions and Open Problems}
\label{sec:conclusions}

In this work, we gave nearly tight lower bounds for private covariance estimation with respect to the Frobenius and spectral norms under the Gaussian distribution and for mean estimation with respect to the $\ell_2$-norm under distributions with bounded second moments.
This constitutes important progress in the lower bounds literature under approximate differential privacy.
On a technical level, we significantly generalized the fingerprinting method.
An obvious question is whether the fingerprinting method can be further generalized and what are its limitations.
A second obvious question is to generalize Theorem~\ref{thm:lb-second-moment} for higher-order moments.

Another question involves expanding upon the results of Theorems~\ref{thm:lb-maha} and~\ref{thm:lb-spectral} to hold in the regimes $\alpha = \cO\paren{\sqrt{d}}$ and $\alpha = \cO\paren{1}$, respectively.
For Frobenius estimation, the condition $\alpha = \cO\paren{1}$ is a consequence of the distribution over covariance matrices that we used to prove the result, whereas for spectral estimation the condition $\alpha = \cO\paren{\frac{1}{\sqrt{d}}}$ is an artifact of our reduction-based approach.
Especially in the latter case, it would be interesting to come up with a ``direct" way to obtain the lower bound, instead of reducing from Frobenius estimation.

Finally, an interesting question may be to obtain lower bounds for the problems considered here in the \emph{high-probability regime}.
Specifically, the guarantees in our lower bounds are expressed in terms of the MSE and in Section~\ref{subsec:results} we sketched the argument to convert them to constant probability guarantees.
It would be interesting to come up with techniques to prove tight lower bounds when the probability of success is $1 - \beta$ for $\beta \in \paren{0, 1}$.
This was done in the context of pure DP in~\cite{HopkinsKM21}.

\section*{Acknowledgments}

The authors would like to thank Jonathan Ullman for helpful feedback on the manuscript, Haoshu Xu for identifying a bug in the proof of Lemma~\ref{lem:upp_bound2} which affected how the range of $\delta$ is identified in Theorems~\ref{thm:lb-maha} and~\ref{thm:gaussian-mean}, as well as the anonymous reviewers at NeurIPS for their comments.
Finally, we would like to thank Yichen Wang for clarifications regarding some aspects of~\cite{CaiWZ23}.

\bibliographystyle{alpha}
\bibliography{biblio.bib}

\appendix
\renewcommand{\thesection}{\Alph{section}}

\section{Facts from Linear Algebra and Geometry}
\label{sec:facts_lin_alg}

In this appendix we collect some results from linear algebra and geometry which are referenced throughout this work.

\begin{definition}[$\eps$-net]
\label{def:net}
Let $\paren{T, d}$ be a metric space.
Consider a subset $K \subset T$ and let $\eps > 0$.
A subset $\cN_{\eps} \subseteq K$ is called an \emph{$\eps$-net} of $K$ if the following holds:
\[
    \forall x \in K, \exists y \in \cN_{\eps} \colon d\paren{x, x_0} \le \eps.
\]
\end{definition}

\begin{definition}[Covering Number]
\label{def:cov_num}
The smallest possible cardinality of an $\eps$-net of $K$ is called the \emph{covering number} of $K$ and denoted by $\cN\paren{K, d, \eps}$. Equivalently, $\cN\paren{K, d, \eps}$ is the smallest number of closed balls with centers in $K$ and radii $\eps$ whose union covers $K$. 
\end{definition}

\begin{fact}[Covering numbers of $\bS^{d - 1}$]
\label{fact:cov_num_unit_sphere}
The covering numbers of the unit Euclidean sphere $\bS^{d - 1}$ satisfy the following for any $\eps > 0$:
\[
    \paren{\frac{1}{\eps}}^d \le \cN\paren{\bS^{d - 1}, \eps} \le \paren{1 + \frac{1}{\eps}}^d.
\]
\end{fact}

\begin{fact}[Exercise $4.4.2$ from~\cite{Vershynin18}]
\label{fact:l2_norm_bound}
Let $x \in \bR^d$ and $\cN_{\eps}$ be an $\eps$-net on $\bS^{d - 1}$.
It holds that:
\[
    \sup\limits_{y \in \cN_{\eps}} \abs{\iprod{x, y}} \le \llnorm{x} \le \frac{1}{1 - \eps} \cdot \sup\limits_{y \in \cN_{\eps}} \abs{\iprod{x, y}}.
\]
\end{fact}

\begin{fact}[Spectral Norm of Kronecker Product]
\label{fact:kronecker_norm}
Let $A \in \bR^{n \times m}$ and $B \in \bR^{k \times \ell}$.
Then, it holds that $\llnorm{A \otimes B} = \llnorm{A} \llnorm{B}$.
\end{fact}

\begin{fact}[Frobenius Norm of Matrix Product]
\label{fact:frob_product}
Let $M, N \in \bR^{d \times d}$.
Then, $\fnorm{M N} \le \llnorm{M} \fnorm{N}$, and $\fnorm{M N} \le \llnorm{N} \fnorm{M}$.
\end{fact}

For a proof of the first of the two inequalities, see Lemma $4.26$ of~\cite{DiakonikolasKKLMS16}.
The proof for the second one is analogous but we have to work with the rows of $M$ instead of the columns of $N$.

\begin{lemma}
\label{lem:mah_rescale1}
Let $\Sigma, A \in \bR^{d \times d}$ be a PSD and a non-singular matrix, respectively, and let $X \in \bR^{d \times d}$.
Then, it holds that $\mnorm{A \Sigma A^{\top}}{X} = \mnorm{\Sigma}{A^{- 1} X \paren{A^{\top}}^{- 1}}$.
\end{lemma}

For a proof of this statement for symmetric $A$ (the proof does not change), we refer the reader to Section $2.2$ of~\cite{Gopi19}.

\begin{lemma}
\label{lem:mah_rescale2}
Let $A, B \in \bR^{d \times d}$ be symmetric positive definite matrices.
Then, it holds that:
\[
    \mnorm{A}{A - A^{\frac{1}{2}} B^{\frac{1}{2}} A^{- 1} B^{\frac{1}{2}} A^{\frac{1}{2}}} = \mnorm{A}{A - B}.
\]
\end{lemma}

\begin{proof}
We have that:
\begin{align*}
    \mnorm{A}{A - A^{\frac{1}{2}} B^{\frac{1}{2}} A^{- 1} B^{\frac{1}{2}} A^{\frac{1}{2}}}^2 = \fnorm{\id - B^{\frac{1}{2}} A^{- 1} B^{\frac{1}{2}}}^2 &= \tr\paren{\paren{\id - B^{\frac{1}{2}} A^{- 1} B^{\frac{1}{2}}}^2} \\
                                                                                                                                                       &= \tr\paren{\id - 2 B^{\frac{1}{2}} A^{- 1} B^{\frac{1}{2}} + B^{\frac{1}{2}} A^{- 1} B A^{- 1} B^{\frac{1}{2}}} \\
                                                                                                                                                       &\overset{\paren{a}}{=} \tr\paren{\id - 2 A^{- \frac{1}{2}} B A^{- \frac{1}{2}} + \paren{A^{- \frac{1}{2}} B A^{- \frac{1}{2}}}^2} \\
                                                                                                                                                       &= \tr\paren{\paren{\id - A^{- \frac{1}{2}} B A^{- \frac{1}{2}}}^2} \\
                                                                                                                                                       &= \fnorm{\id - A^{- \frac{1}{2}} B A^{- \frac{1}{2}}}^2 \\
                                                                                                                                                       &= \mnorm{A}{A - B}^2,
\end{align*}
where in $\paren{a}$ we use the cyclic property of the trace in individual terms.
\end{proof}

\begin{theorem}[Gershgorin Circle Theorem (see~\cite{Gershgorin31})]
\label{thm:gershg}
Let $A \in \bR^{d \times d}$.
For any eigenvalue $\lambda$ of $A$, there exists an $i \in \brk{d}$, such that $\abs{\lambda - a_{i i}} \le \sum\limits_{j \neq i} \abs{a_{i j}}$.
\end{theorem}

\section{Facts from Probability \& Statistics}
\label{sec:facts_prob_stats}

In this appendix we collect some results from probability and statistics which are referenced throughout this work.

\begin{fact}
\label{fact:exp_fam_bernoulli}
The Bernoulli distribution $\Be\paren{p}$ is an exponential family with support $S = \brc{0, 1}$ and $h\paren{x} = 1, \eta = \ln\paren{\frac{1 - p}{p}}, T\paren{x} = x, Z\paren{\eta} = \ln\paren{1 + e^{\eta}}$.
\end{fact}

\begin{fact}
\label{fact:exp_fam_norm_mean}
The $d$-dimensional Gaussian distribution with unit covariance and unknown mean $\cN\paren{\mu, \id}$ is an exponential family with support $S = \bR^d$ and $h\paren{x} = \frac{1}{\paren{2 \pi}^{\frac{d}{2}}} e^{- \frac{\llnorm{x}^2}{2}}, \eta = \mu, T\paren{x} = x, Z\paren{\eta} = \frac{\llnorm{\eta}^2}{2}$.
\end{fact}

\begin{fact}
\label{fact:exp_fam_norm_cov}
The $d$-dimensional Gaussian distribution with mean $0$ and unknown covariance $\cN\paren{0, \Sigma}$ is an exponential family with support $S = \bR^d$ and $h\paren{x} = 1, \eta = \paren{\Sigma^{- 1}}^{\flat}, T\paren{x} = - \frac{1}{2} \paren{x x^{\top}}^{\flat}, Z\paren{\eta} = \frac{d}{2} \ln\paren{2 \pi} - \frac{1}{2} \ln\paren{\det\paren{\eta^{\#}}}$.
\end{fact}

\begin{fact}[Cauchy-Schwarz]
\label{fact:cs}
We consider the following variants of the Cauchy-Schwarz inequality:
\begin{enumerate}
    \item For every $x, y \in \bR^d$, we have $\abs{\iprod{x, y}} \le \llnorm{x} \llnorm{y}$.
    \item Let $X, Y$ be random variables over $\bR$.
    Then:
    \[
        \ex{}{\abs{X Y}} \le \sqrt{\ex{}{X^2}} \sqrt{\ex{}{Y^2}},
    \]
    assuming all the quantities involved are well-defined.
\end{enumerate}
\end{fact}

\begin{fact}
\label{fact:erfc}
We denote the \emph{complementary error function} by:
\[
    \erfc\paren{x} \coloneqq \frac{2}{\sqrt{\pi}} \int\limits_x^{\infty} e^{- t^2} \, dt, x \geq 0.
\]
It holds that $\erfc\paren{x} \le e^{- x^2}, \forall x \geq 0$.
\end{fact}

\begin{fact}[Lemma 1 of~\cite{LaurentM00}]
\label{fact:hd-gaussian-tail}
If $X \sim \chi^2\paren{k}$, and $\beta \in \brk{0, 1}$, then:
\[
    \pr{}{X - k \geq 2 \sqrt{k \ln\paren{\frac{1}{\beta}}} + 2 \ln\paren{\frac{1}{\beta}}} \le \beta,
\]
and:
\[
    \pr{}{k - X \geq 2 \sqrt{k\ln\paren{\frac{1}{\beta}}}} \le \beta.
\]
Equivalently, the above can be written as:
\[
    \pr{}{X \geq t} \le e^{- \frac{\paren{\sqrt{2 t - k} - \sqrt{k}}^2}{4}}, \forall t \geq k,
\]
and:
\[
    \pr{}{X \le t} \le e^{- \frac{\paren{k - t}^2}{4}}, \forall t \le k.
\]
Thus, if $Y \sim \cN\paren{0, \id}$, then $\pr{}{\llnorm{Y}^2 \geq d + 2 \sqrt{d \ln\paren{\frac{1}{\beta}}}+ 2 \ln\paren{\frac{1}{\beta}}} \le \beta$.
\end{fact}

\begin{fact}[Gaussian Hanson-Wright Inequality - Theorem $1$ and Proposition $1$ from~\cite{Moshksar21}]
\label{fact:hanson-wright}
Let $X \in \bR^d$ be a random vector with $X \sim \cN\paren{0, \id}$.
For any non-zero symmetric matrix $A \in \bR^{d \times d}$:
\[
    \pr{}{\abs{X^{\top} A X - \ex{}{X^{\top} A X}} \geq t} \le 2 \exp\paren{- c \min\brc{\frac{t^2}{\fnorm{A}^2}, \frac{t}{\llnorm{A}}}}, \forall t \geq 0,
\]
where $c$ is an absolute constant with $c \approx 0.036425$.
\end{fact}

\section{Basic Applications: Recovering Existing Lower Bounds}
\label{sec:existing}

We show how to use Theorem~\ref{thm:lower_bound} to recover existing lower bounds from \cite{KamathLSU19}:
$\paren{1}$ mean estimation of binary product distributions,
and $\paren{2}$ mean estimation of high-dimensional Gaussians.
Note that both classes of distributions are exponential families.
The error metric for these lower bounds would be the mean-squared-error (MSE), as opposed to constant-probability-error, as in \cite{KamathLSU19}, but the bounds could be converted to constant-probability bounds at the cost of a $\log\paren{d}$ factor in the sample complexity (as was sketched in Section~\ref{subsec:results}).

\subsection{Mean Estimation of Binary Product Distributions}
\label{subsec:bin_prod}

We start by stating the theorem for mean estimation of binary product distributions over $\brc{0,1}^d$.

\begin{theorem}[Product Distribution Mean Estimation]
\label{thm:product-dist-mean}
There exists a distribution $\cD$ over vectors $p \in \brk{\frac{1}{3}, \frac{2}{3}}^d$ such that, given $p \sim \cD$ and $n$ independent samples $X \coloneqq \paren{X_1, \dots, X_n}$ from a binary product distribution $P$ over $\brc{0, 1}^d$ with mean $p$, for any $\alpha = \cO\paren{\sqrt{d}}$ and any $\paren{\eps, \delta}$-DP mechanism $M \colon \brc{0, 1}^{n \times d} \to \brk{\frac{1}{3}, \frac{2}{3}}^d$ with $\eps, \delta \in \brk{0, 1}$, and $\delta = \cO\paren{\frac{1}{n}}$ that satisfies $\ex{X, M}{\llnorm{M\paren{X} - p}^2} \le \alpha^2$, it holds that $n = \Omega\paren{\frac{d}{\alpha \eps}}$.
\end{theorem}

Fact~\ref{fact:exp_fam_bernoulli} establishes that Bernoulli distributions are an exponential family.
Using that, we have that the probability mass function of binary product distributions can be written as:
\begin{align*}
    p_{\eta}\paren{x} &= h\paren{x} \exp\paren{\eta^{\top} T\paren{x} - Z\paren{\eta}}, \forall x \in \brc{0, 1}^d, \\
           h\paren{x} &= 1, \\
           T\paren{x} &= x, \\
                 \eta &= \paren{\ln\paren{\frac{1 - p_1}{p_1}}, \dots, \ln\paren{\frac{1 - p_d}{p_d}}}^{\top}, \\
        Z\paren{\eta} &= \prod\limits_{j \in \brk{d}} \ln\paren{1 + e^{\eta_j}}.
\end{align*}
Our process to generate $\eta$ involves drawing independently $\eta_j \sim \cU\brk{\pm \ln\paren{2}}$.
We have that $\eta_j = \ln\paren{\frac{1}{p_j} - 1} \iff p_j = \frac{1}{1 + e^{\eta_j}}$, yielding $\eta_j \in \brk{\pm \ln\paren{2}} \iff p_j \in \brk{\frac{1}{3}, \frac{2}{3}}$.
Thus, we have $I_j = \brk{\pm \ln\paren{2}}, R_j = 2 \ln\paren{2}, \forall j \in \brk{d} \implies \llnorm{R}^2 = 4 \ln^2\paren{2} d, \infnorm{R} = 2 \ln\paren{2}$, and $m = 0$.

We now show how to reduce estimating $\eta$ with an $\ell_2$-guarantee under $\paren{\eps, \delta}$-DP to estimating $p$ with an $\ell_2$-guarantee under the same constraint.

\begin{lemma}
\label{lem:reduction_bin_prod}
Let $p \in \brk{\frac{1}{3}, \frac{2}{3}}^d$ be a randomly generated vector and let $X \sim P^{\bigotimes n}$ be a dataset drawn from a binary product distribution over $\brc{0, 1}^d$ with mean $p$.
If $M \colon \brc{0, 1}^{n \times d} \to \brk{\frac{1}{3}, \frac{2}{3}}^d$ is an $\paren{\eps, \delta}$-DP mechanism satisfying $\ex{X, M}{\llnorm{M\paren{X} - p}^2} \le \alpha^2 \le \frac{d}{9}$, then there exists a $T_M \colon \brc{0, 1}^{n \times d} \to \brk{\pm \ln\paren{2}}^d$ that is also $\paren{\eps, \delta}$-DP and satisfies $\ex{X, T_M}{\llnorm{T_M\paren{X} - \eta}^2} \le \frac{81}{4} \alpha^2$.
\end{lemma}

\begin{proof}
We define $T_M$ to be $T_{M, j}\paren{X} = \ln\paren{\frac{1 - M_j\paren{X}}{M_j\paren{X}}}, \forall j \in \brk{d}$.
$T_M$ satisfies $\paren{\eps, \delta}$-DP by Lemma~\ref{lem:post_processing}, and the only randomness in $T_M$ is that of $M$.
Consider now the function $g \colon \paren{0, 1} \to \bR$ with $g\paren{x} = \ln\paren{\frac{1}{x} - 1}$ and $g'\paren{x} = - \frac{1}{x \paren{1 - x}}$.
We have $\eta_j = g\paren{p_j}$ and $T_{M, j}\paren{X} = g\paren{T_j\paren{X}}$.
By the Mean Value Theorem, we have for some $\xi_j$ between $\eta_j$ and $M_j\paren{X}$ (implying that $\xi_j \in \brk{\frac{1}{3}, \frac{2}{3}}$):
\[
    \abs{g\paren{M_j\paren{X}} - g\paren{p_j}} = \abs{g'\paren{\xi_j}} \abs{M_j\paren{X} - p_j} \le \frac{9}{2} \abs{M_j\paren{X} - p_j}.
\]
Applying this coordinate-wise, we get $\ex{X, T_M}{\llnorm{T_M\paren{X} - \eta}^2} \le \frac{81}{4} \ex{X, M}{\llnorm{M\paren{X} - p}^2} \le \frac{81}{4} \alpha^2$.
\end{proof}

We now present a restatement of the previous lemma from the point of view of lower bounds.

\begin{corollary}
\label{cor:red_bin_prod_lb}
Let $p \in \brk{\frac{1}{3}, \frac{2}{3}}^d$ be a randomly generated vector and let $X \sim P^{\bigotimes n}$ be a dataset drawn from a binary product distribution over $\brc{0, 1}^d$ with mean $p$.
If any $\paren{\eps, \delta}$-DP mechanism $T \colon \brc{0, 1}^{n \times d} \to \brk{\pm \ln\paren{2}}^d$ satisfying $\ex{X, T}{\llnorm{T\paren{X} - \eta}^2} \le \frac{81}{4} \alpha^2 \le \frac{9}{4} d$ requires at least $n \geq n_{\eta}$ samples, the same sample complexity lower bound holds for any $\paren{\eps, \delta}$-DP mechanism $M \colon \brc{0, 1}^{n \times d} \to \brk{\frac{1}{3}, \frac{2}{3}}^d$ that satisfies $\ex{X, M}{\llnorm{M\paren{X} - p}^2} \le \alpha^2$.
\end{corollary}

In contrast to other lower bounds proven in this work using Theorem~\ref{thm:lower_bound}, we do not need to study the concentration properties of $\llnorm{T\paren{X_i} - \mu_T} = \llnorm{X_i - p}$.
Indeed, since $T\paren{X} = X \in \brc{0, 1}^d$, it suffices to pick a large enough threshold $T > 0$ so that $\pr{X_i}{\llnorm{X_i - p} > \frac{t}{2 \ln^3\paren{2} \sqrt{d}}} = 0, \forall t \geq T$.

Based on the above, we are now ready to prove the main theorem of this section.

\begin{proof}[Proof of Theorem~\ref{thm:product-dist-mean}]
Let $M \colon \brc{0, 1}^{n \times d} \to \brk{\pm \ln\paren{2}}^d$ be an $\paren{\eps, \delta}$-DP mechanism with:
\[
    \ex{X, M}{\llnorm{M\paren{X} - \eta}^2} \le \frac{81}{4} \alpha^2 \le \frac{\ln^2\paren{2} d}{6} = \frac{\llnorm{R}^2}{24}.
\]
By Theorem~\ref{thm:lower_bound} for product distributions with $\eta_j \in \brk{\pm \ln\paren{2}}, \forall j \in \brk{d}$, we get:
\begin{align}
    n \brc{2 \delta T + 2 \eps \ex{\eta}{\sqrt{\ex{X_{\sim i}, M}{s^{\top} \Sigma_T s}}} + 2 \ex{\eta}{\int\limits_T^{\infty} \pr{X_i}{\llnorm{X_i - p} > \frac{t}{2 \ln^3\paren{2} \sqrt{d}}} \, dt}} \geq \frac{\ln^2\paren{2} d}{6}. \label{eq:prod1}
\end{align}
Observe that, since we have a product distribution and $T\paren{x} = x$, it holds that $\mu_T = p$ and:
\[
    \Sigma_T = \Sigma = \diag\brc{p_1 \paren{1 - p_1}, \dots, p_d \paren{1 - p_d}} \preceq \frac{1}{4} \id,
\]
for $p_i \in \brk{\frac{1}{3}, \frac{2}{3}}, \forall i \in \brk{d}$.

Thus, we have:
\begin{align}
    \ex{X_{\sim i}, M}{s^{\top} \Sigma_T s} = \frac{1}{4} \ex{X_{\sim i}, M}{\llnorm{s}^2} \le \frac{\infnorm{R}^4}{64} \alpha^2 = \frac{\ln^4\paren{2} \alpha^2}{4}. \label{eq:prod2}
\end{align}
Furthermore, we remark that $\llnorm{X_i - p} \le \frac{2}{3} \sqrt{d}$.
Thus, setting $T = \frac{4 \ln^3\paren{2} d}{3}$ yields:
\begin{align}
    \int\limits_T^{\infty} \pr{X_i}{\llnorm{X_i - p} > \frac{t}{2 \ln^3\paren{2} \sqrt{d}}} \, dt = 0. \label{eq:prod3}
\end{align}
Substituting (\ref{eq:prod2}) and (\ref{eq:prod3}) into (\ref{eq:prod1}), we get:
\[
    n \paren{\frac{8 \ln^3\paren{2} d \delta}{3} + \frac{\ln^2\paren{2} \alpha \eps}{2}} \geq \frac{\ln^2\paren{2} d}{6}.
\]
Setting $\delta \le \frac{1}{32 \ln\paren{2} n}$ results in $n \geq \frac{d}{12 \alpha \eps}$, so appealing to Corollary~\ref{cor:red_bin_prod_lb} completes the proof.
\end{proof}

\subsection{Mean Estimation of High-Dimensional Gaussians}

Our techniques also apply to recovering lower bounds for mean estimation of a Gaussian with known covariance.
By Fact~\ref{fact:exp_fam_norm_mean}, we have that this class of distributions is an exponential family with $\eta = \mu$ and $T\paren{x} = x$ (implying that $\mu_T = \mu, \Sigma_T = \Sigma$).
Therefore, we do not need to resort to a reduction-based approach (as we did with binary product distributions).

\begin{theorem}[Gaussian Mean Estimation]\label{thm:gaussian-mean}
Given $\mu \sim \cU\paren{\brk{\pm 1}^d}$ and $X \sim \cN\paren{\mu, \bI}^{\bigotimes n}$, for any $\alpha = \cO\paren{\sqrt{d}}$ and any $\paren{\eps, \delta}$-DP mechanism $M \colon \bR^{n \times d} \to \brk{\pm 1}^d$ with $\eps, \delta \in \brk{0, 1}$, and $\delta \le \cO\paren{\min\brc{\frac{1}{n}, \frac{\sqrt{d}}{n \sqrt{\log\paren{\frac{n}{\sqrt{d}}}}}}}$ that satisfies $\ex{X, M}{\llnorm{M\paren{X} - \mu}^2} \le \alpha^2$, it holds that $n = \Omega\paren{\frac{d}{\alpha \eps}}$.
\end{theorem}

Before proving Theorem~\ref{thm:gaussian-mean}, we start with a lemma that is concerned with upper-bounding the term involving the tail probability of $\llnorm{T\paren{X_i} - \mu_T}$.
The lemma does not rely on the approach of Remark~\ref{rem:exp_chernoff}, instead using the concentration properties of the $\chi^2$-distribution (Fact~\ref{fact:hd-gaussian-tail}).

\begin{lemma}
\label{lem:tail_prob_mean_ub1}
Let $\mu \in \brk{\pm 1}^d$, and $X_i \sim \cN\paren{\mu, \id}$.
Assuming that $T \geq 2 d$, we have:
\[
    \int\limits_T^{\infty} \pr{X_i}{\llnorm{T\paren{X_i} - \mu_T} > \frac{t}{2 \sqrt{d}}} \, dt \le \sqrt{2 \pi d} e^{- \frac{\paren{\sqrt{T^2 - 2 d^2} - \sqrt{2} d}^2}{8 d}}.
\]
\end{lemma}

\begin{proof}
Since $X_i \sim \cN\paren{\mu, \id}$, we have $T\paren{X_i} - \mu_T = X_i - \mu$.
We observe that $X_i - \mu \sim \cN\paren{0, \id} \implies \llnorm{X_i - \mu}^2 \sim \chi^2\paren{d}$.
Since $T \geq 2 d \implies \frac{t^2}{4 d} \geq \frac{T^2}{4 d} \geq d$, we can apply Fact~\ref{fact:hd-gaussian-tail} with $k = d$, yielding:
\begin{align*}
    \int\limits_T^{\infty} \pr{X_i}{\llnorm{X_i - \mu}^2 > \frac{t^2}{4 d}} \, dt &\le \bigintsss\limits_T^{\infty} \exp\brk{- \frac{1}{4} \paren{\sqrt{\frac{t^2}{2 d} - d} - \sqrt{d}}^2} \, dt \\
       \overset{z^2 = \frac{1}{4} \paren{\sqrt{\frac{t^2}{2 d} - d} - \sqrt{d}}^2}&{=} 2 \sqrt{d} \bigintss\limits_{\frac{\sqrt{\frac{T^2}{2 d} - d} - \sqrt{d}}{2}}^{\infty} \frac{2 z + \sqrt{d}}{\sqrt{z^2 + \paren{z + \sqrt{d}}^2}} e^{- z^2} \, dz \\
                                                               \overset{\paren{a}}&{\le} 2 \sqrt{2 d} \int\limits_{\frac{\sqrt{\frac{T^2}{2 d} - d} - \sqrt{d}}{2}}^{\infty} e^{- z^2} \, dz \\
                                                                                  &= \sqrt{2 \pi d} \erfc\paren{\frac{\sqrt{\frac{T^2}{2 d} - d} - \sqrt{d}}{2}} \\
                                                                                  \overset{\paren{b}}&{\le} \sqrt{2 \pi d} e^{- \frac{\paren{\sqrt{\frac{T^2}{2 d} - d} - \sqrt{d}}^2}{4}} \\
                                                                                  &= \sqrt{2 \pi d} e^{- \frac{\paren{\sqrt{T^2 - 2 d^2} - \sqrt{2} d}^2}{8 d}},
\end{align*}
where $\paren{a}$ is by the inequality $a + b \le \sqrt{2} \sqrt{a^2 + b^2}$ for $a = z, b = z + \sqrt{d}$, and $\paren{b}$ is by Fact~\ref{fact:erfc}.
\end{proof}

We are now ready to prove the main theorem of this section.

\begin{proof}[Proof of Theorem~\ref{thm:gaussian-mean}]
Let $M \colon \bR^{n \times d} \to \brk{\pm 1}^d$ be an $\paren{\eps, \delta}$-DP mechanism with:
\[
    \ex{X, M}{\llnorm{M\paren{X} - \mu}^2} \le \alpha^2 \le \frac{d}{6} = \frac{\llnorm{R}^2}{24}.
\]
By Theorem~\ref{thm:lower_bound} for Gaussians $\cN\paren{\mu, \id}$ with $\mu_j \in \brk{\pm 1}, \forall j \in \brk{d}$, we get:
\begin{align}
    n \brc{2 \delta T + 2 \eps \ex{\mu}{\sqrt{\ex{X_{\sim i}, M}{\llnorm{s}^2}}} + 2 \ex{\mu}{\int\limits_T^{\infty} \pr{X_i}{\llnorm{X_i - \mu} > \frac{t}{2 \sqrt{d}}} \, dt}} \geq \frac{d}{6}. \label{eq:gauss1}
\end{align}
We have:
\begin{align}
    \ex{X_{\sim i}, M}{\llnorm{s}^2} \le \frac{\infnorm{R}^4}{16} \alpha^2 = \alpha^2. \label{eq:gauss2}
\end{align}
We assume that $T \geq 2 d$, as is required by Lemma~\ref{lem:tail_prob_mean_ub1}.
Later, picking a specific value for $T$, we will verify that this condition is satisfied.
Substituting based on that lemma and (\ref{eq:gauss2}) into (\ref{eq:gauss1}), we get:
\begin{align}
    n \paren{2 \delta T + 2 \alpha \eps + 2 \sqrt{2 \pi d} e^{- \frac{\paren{\sqrt{T^2 - 2 d^2} - \sqrt{2} d}^2}{8 d}}} \geq \frac{d}{6}. \label{eq:mean_est_ineq1}
\end{align}
It remains to set the values of $T$ and $\delta$ appropriately so that:
\[
    \delta T \geq 2 \sqrt{2 \pi d} e^{- \frac{\paren{\sqrt{T^2 - 2 d^2} - \sqrt{2} d}^2}{8 d}} \text{ and } 3 n \delta T \le \frac{d}{12}.
\]
The first of these two conditions yields:
\[
    T e^{\frac{\paren{\sqrt{T^2 - 2 d^2} - \sqrt{2} d}^2}{8 d}} \geq \frac{2 \sqrt{2 \pi d}}{\delta}.
\]
We set $T = 2 \sqrt{d} \sqrt{\ln\paren{\frac{1}{\delta}} + \paren{\sqrt{\ln\paren{\frac{1}{\delta}}} + \sqrt{d}}^2}$.
Since $\delta \le 1$, we have $\ln\paren{\frac{1}{\delta}} \geq 0 \implies T \geq 2 d$, thus satisfying the condition of Lemma~\ref{lem:tail_prob_mean_ub1}.
Substituting to the previous yields:
\begin{align*}
    2 \sqrt{d} \sqrt{\ln\paren{\frac{1}{\delta}} + \paren{\sqrt{\ln\paren{\frac{1}{\delta}}} + \sqrt{d}}^2} \cdot \frac{1}{\delta} \geq \frac{2 \sqrt{2 \pi d}}{\delta} &\iff \sqrt{\ln\paren{\frac{1}{\delta}} + \paren{\sqrt{\ln\paren{\frac{1}{\delta}}} + \sqrt{d}}^2} \geq \sqrt{2 \pi} \\
    &\iff \ln\paren{\frac{1}{\delta}} + \sqrt{d} \sqrt{\ln\paren{\frac{1}{\delta}}} + \frac{1}{2} \paren{d - 2 \pi} \geq 0.
\end{align*}
Setting $x \coloneqq \sqrt{\ln\paren{\frac{1}{\delta}}} \geq 0$, the above is a quadratic inequality with respect to $x$.
Depending on whether the trinomial in the LHS has any positive roots, this might imply a constraint on $\delta$.
The discriminant is $\Delta = 4 \pi - d$, so taking cases based on $d$, the strictest constraint we get is $\delta \le e^{- \pi + \frac{\sqrt{4 \pi - 1}}{2}} \approx 0.23$, which is for $d = 1$.
Respecting this constraint, we now must identify a range of values for $\delta$ so that we have $3 n \delta T \le \frac{d}{12}$.
We have:
\[
    3 n \delta T \le \frac{d}{12} \iff 6 n \delta \sqrt{d} \sqrt{\ln\paren{\frac{1}{\delta}} + \paren{\sqrt{\ln\paren{\frac{1}{\delta}}} + \sqrt{d}}^2} \le \frac{d}{12}.
\]
It holds that $\ln\paren{\frac{1}{\delta}} + \paren{\sqrt{\ln\paren{\frac{1}{\delta}}} + \sqrt{d}}^2 \le 2 \paren{\sqrt{\ln\paren{\frac{1}{\delta}}} + \sqrt{d}}^2$, so it suffices to satisfy:
\begin{align}
    6 \sqrt{2} n \delta \sqrt{d} \paren{\sqrt{\ln\paren{\frac{1}{\delta}}} + \sqrt{d}} \le \frac{d}{12} \iff \delta \paren{1 + \sqrt{\frac{\ln\paren{\frac{1}{\delta}}}{d}}} \le \frac{1}{72 \sqrt{2} n}. \label{eq:delta_condition}
\end{align}
The previous inequality imposes a condition on $\delta$.
Indeed, to ensure that $\sqrt{\frac{\ln\paren{\frac{1}{\delta}}}{d}} \geq 0$, it is necessary that $\frac{1}{72 \sqrt{2} n \delta} - 1 \geq 0 \iff \delta \le \frac{1}{72 \sqrt{2} n}$, which is stronger than $\delta \le 0.23$.

For (\ref{eq:delta_condition}) to hold, we will show that it suffices to have:
\[
    \delta \le \min\brc{\frac{1}{144 \sqrt{2} n}, \frac{\sqrt{d}}{288 \sqrt{2} n \sqrt{\ln\paren{\frac{144 \sqrt{2} n}{\sqrt{d}}}}}},
\]
which trivially satisfies the previous constraint.

We set $\phi \coloneqq \frac{\sqrt{d}}{144 \sqrt{2} n}$, so the proposed range for $\delta$ becomes $\delta \le \min\brc{\frac{\phi}{\sqrt{d}}, \frac{\phi}{2 \sqrt{\ln\paren{\frac{1}{\phi}}}}}$.
Additionally, (\ref{eq:delta_condition}) can be equivalently written as:
\begin{align}
    \delta \paren{1 + \sqrt{\frac{\ln\paren{\frac{1}{\delta}}}{d}}} \le \frac{2 \phi}{\sqrt{d}}. \label{eq:delta_phi}
\end{align}
Observe that, depending on how $\delta$ compares with $e^{- d}$ determines which of the terms $1$ and $\sqrt{\frac{\ln\paren{\frac{1}{\delta}}}{d}}$ dominates.
Thus, our strategy to verify our claim is by considering cases based on whether the proposed range for $\delta$ includes $e^{- d}$.
This yields:
\begin{enumerate}
    \item $\min\brc{\frac{\phi}{\sqrt{d}}, \frac{\phi}{2 \sqrt{\ln\paren{\frac{1}{\phi}}}}} < e^{- d}$.
    We start by arguing that:
    \begin{align}
        \min\brc{\frac{\phi}{\sqrt{d}}, \frac{\phi}{2 \sqrt{\ln\paren{\frac{1}{\phi}}}}} < e^{- d} \iff \frac{\phi}{2 \sqrt{\ln\paren{\frac{1}{\phi}}}} < e^{- d}. \label{eq:equivalence}
    \end{align}
    The $\impliedby$ direction is trivial, so we focus on the $\implies$ direction.
    
    We assume that $\min\brc{\frac{\phi}{\sqrt{d}}, \frac{\phi}{2 \sqrt{\ln\paren{\frac{1}{\phi}}}}} < e^{- d}$ but $\frac{\phi}{2 \sqrt{\ln\paren{\frac{1}{\phi}}}} \geq e^{- d}$.
    Then, it must be the case that $\frac{\phi}{\sqrt{d}} < e^{- d}$.
    Setting $y \coloneqq \frac{1}{\phi}$, the above system of inequalities can be written in the form:
    \[
        \begin{cases}
            \frac{e^{2 d}}{4} \geq y^2 \ln\paren{y} \\
            y > \frac{e^d}{\sqrt{d}}
        \end{cases}.
    \]
    The function $y^2 \ln\paren{y}$ is increasing for $y \geq e^{- \frac{1}{2}}$.
    We have $\frac{e^d}{\sqrt{d}} > e^{- \frac{1}{2}}, \forall d \geq 1$.
    Thus, for $y > \frac{e^d}{\sqrt{d}}$, we get:
    \[
        \frac{e^{2 d}}{4} \geq y^2 \ln\paren{y} > \frac{e^{2 d}}{d} \ln\paren{\frac{e^d}{\sqrt{d}}} \implies 3 d < 2 \ln\paren{d}.
    \]
    However, this leads to a contradiction, because $3 d < 2 \ln\paren{d} \le 2 d - 2 \implies d < - 2$.
    
    Consequently, we have shown (\ref{eq:equivalence}).
    Since we have assumed that $\min\brc{\frac{\phi}{\sqrt{d}}, \frac{\phi}{2 \sqrt{\ln\paren{\frac{1}{\phi}}}}} < e^{- d}$, it must be the case that $\frac{\phi}{2 \sqrt{\ln\paren{\frac{1}{\phi}}}} < e^{- d}$.
    
    We now turn our attention again to (\ref{eq:delta_phi}).
    We remark that, for $\delta < e^{- d}$, we get $\sqrt{\frac{\ln\paren{\frac{1}{\delta}}}{d}} > 1$.
    As a result, in order to satisfy (\ref{eq:delta_phi}), it suffices to have:
    \begin{align}
        2 \delta \sqrt{\frac{\ln\paren{\frac{1}{\delta}}}{d}} \le \frac{2 \phi}{\sqrt{d}} \iff \delta \sqrt{\ln\paren{\frac{1}{\delta}}} \le \phi. \label{eq:delta_phi1}
    \end{align}
    Observe that the function $\delta \sqrt{\ln\paren{\frac{1}{\delta}}}$ is increasing for $\delta < e^{- d} < e^{- \frac{1}{2}}$.
    Combined with our previous remarks, this implies that, in order to verify (\ref{eq:delta_phi1}) for the values of $\delta$ we picked, it suffices to show that it holds for $\delta = \frac{\phi}{2 \sqrt{\ln\paren{\frac{1}{\phi}}}}$.
    So, we have to verify the inequality:
    \[
        \paren{\frac{\phi}{2 \sqrt{\ln\paren{\frac{1}{\phi}}}}} \sqrt{\ln\paren{\frac{2 \sqrt{\ln\paren{\frac{1}{\phi}}}}{\phi}}} \le \phi \iff \frac{\ln\paren{2 \sqrt{\ln\paren{\frac{1}{\phi}}}}}{\ln\paren{\frac{1}{\phi}}} \le 3.
    \]
    Setting $z \coloneqq \ln\paren{\frac{1}{\phi}}$, this becomes equivalent to $\frac{\ln\paren{2 \sqrt{z}}}{z} \le 3$, which holds for all $z > 0 \iff \ln\paren{\frac{1}{\phi}} > 0 \iff \phi < 1$.
    This last condition is satisfied by all $\phi$ such that $\frac{\phi}{2 \sqrt{\ln\paren{\frac{1}{\phi}}}} < e^{- d}$, so the desired result has been established.
    \item $\min\brc{\frac{\phi}{\sqrt{d}}, \frac{\phi}{2 \sqrt{\ln\paren{\frac{1}{\phi}}}}} \geq e^{- d}$.
    For this case, we consider two sub-cases, depending on how $\delta$ compares with $e^{- d}$.
    We have:
    \begin{itemize}
        \item $\delta < e^{- d}$.
        As before, we argue that it suffices to have $\delta \sqrt{\ln\paren{\frac{1}{\delta}}} \le \phi$.
        Since $\delta \sqrt{\ln\paren{\frac{1}{\delta}}}$ is increasing for $\delta < e^{- d} < e^{- \frac{1}{2}}$, all we have to do is show that $e^{- d} \sqrt{d} \le \phi$.
        This, however, is an immediate consequence of our assumption that $\min\brc{\frac{\phi}{\sqrt{d}}, \frac{\phi}{2 \sqrt{\ln\paren{\frac{1}{\phi}}}}} \geq e^{- d}$.
        \item $\delta \geq e^{- d}$.
        This implies $\frac{\ln\paren{\frac{1}{\delta}}}{d} \le 1$.
        Thus, for (\ref{eq:delta_phi}) to hold, it suffices that $2 \delta \le \frac{2 \phi}{\sqrt{d}} \iff \delta \le \frac{\phi}{\sqrt{d}}$.
        This is satisfied, because our proposed range for $\delta$ is $\delta \le \min\brc{\frac{\phi}{\sqrt{d}}, \frac{\phi}{2 \sqrt{\ln\paren{\frac{1}{\phi}}}}}$.
    \end{itemize}
\end{enumerate}
We have established that our proposed values of $\delta$ and $T$ imply:
\[
    \delta T \geq 2 \sqrt{2 \pi d} e^{- \frac{\paren{\sqrt{T^2 - 2 d^2} - \sqrt{2} d}^2}{8 d}} \text{ and } 3 n \delta T \le \frac{d}{12},
\]
while respecting all the constraints.

We substitute this to (\ref{eq:mean_est_ineq1}) and get $n \geq \frac{d}{24 \alpha \eps}$, which yields the desired result.
\end{proof}

\end{document}